\documentclass[twocolumn,journal]{IEEEtran}

\usepackage{listings}
\usepackage{amsmath,amssymb,array}
\usepackage{siunitx}
\usepackage{longtable,tabularx}
\usepackage{multirow}
\usepackage{color}
\usepackage[pdftex]{graphicx}
\usepackage{dcolumn}
\usepackage{bm}
\usepackage{algorithm}
\usepackage{algpseudocode}
\usepackage{lipsum}

\DeclareMathOperator{\Tr}{\mathrm{Tr}}

\newtheorem{theorem}{Theorem}

\newtheorem{lemma}[theorem]{Lemma}%
\newtheorem{remark}{Remark}%
\newtheorem{proposition}[theorem]{Propositon}%

\newcommand{\argmin}{\mathop{\rm argmin}\limits}

\newcommand{\frS}{\mathfrak{S}}

\def\mix{\mathop{\rm mix}}

\def\sgn{\mathop{\rm sgn}}
\def\SU{\mathop{\rm SU}}

\def\cH{{\mathcal{H}}}

\def\QED{\mbox{\rule[0pt]{1.5ex}{1.5ex}}}

\newcommand{\qed}{\hfill \QED}

 \newenvironment{proofof}[1]{\vspace*{5mm} \par \noindent
{\it Proof of #1:\hspace{2mm}}}{\qed
}
\def\Label#1{\label{#1}\ [\ \text{#1}\ ]\ }
\def\Label{\label}

\begin{document}

\title{Entanglement measures for detectability}

\author{Masahito Hayashi \IEEEmembership{Fellow, IEEE}
and Yuki Ito 
\thanks{The work of MH was supported 
in part by
the National Natural Science Foundation of
China (Grant No. 62171212).
}
\thanks{Masahito Hayashi is with 
School of Data Science, The Chinese University of Hong Kong,
Shenzhen, Longgang District, Shenzhen, 518172, China
International Quantum Academy, Futian District, Shenzhen 518048, China,
and
the Graduate School of Mathematics, Nagoya University, Nagoya, 464-8602, Japan
(e-mail:hmasahito@cuhk.edu.cn).
Yuki Ito is with 
the Graduate School of Mathematics, Nagoya University, Nagoya, 464-8602, Japan
(e-mail: ito.yuki.g7@s.mail.nagoya-u.ac.jp).} }
\markboth{M. Hayashi 
and Y. Ito: Entanglement measures for detectability}{}

\maketitle

\begin{abstract}
We propose new entanglement measures as the detection performance based on the hypothesis testing setting. We clarify how our measures work for detecting an entangled state by extending the quantum Sanov theorem.
Our analysis covers the finite-length setting. Exploiting this entanglement measure, we present how to derive an entanglement witness to detect the given entangled state by using the geometrical structure of this measure. We derive their calculation formulas for maximally correlated states, and propose their algorithms that work for general entangled states. In addition, we investigate how our algorithm works for solving the membership problem for separability.
Further, employing this algorithm, we propose 
a method to find entanglement witness for a given entangled state. 
\end{abstract}

\begin{IEEEkeywords} 
entanglement measure, entanglement detection, 
membership problem, hypothesis testing,
entanglement witness
\end{IEEEkeywords}


\section{Introduction}
Entanglement is a valuable resource for quantum information processing.
It is important to distinguish a given entangled state from 
separable states, i.e., non-entangled states.
However, detecting an entanglement state
has two meanings, the physical detection and the mathematical detection.
The physical detection means the experimental
detection of a generated entangled state.
Given a density matrix $\rho_{AB}$ being entangled,
when a physical system is prepared as the density matrix $\rho_{AB}$,
we verify whether the generated physical state is entangled.
This problem is a special case of quantum hypothesis testing.
In the mathematical detection,
given a density matrix $\rho_{AB}$, we study whether 
$\rho_{AB}$ is an entangled state or a separable state.
Since the mathematical detection
is often called the membership problem of separability,
in this paper, we consider that 
entanglement detection means the physical one
while this paper addresses these two different problems.
This paper aims to propose entangled measures,
and to present that
our entangled measures work for both topics.

Several practical methods to achieve 
entanglement detection have been proposed \cite{Guhne,BYE,AGD,Dimic,Saggio,Saggio2,CDKKM},
but they did not discuss the optimal performance for the detection.
On the other hand, to characterize the amount of entanglement,
many entangled measures have been proposed.
However, these proposed entangled measures have less relation with the experimental detectability of entanglement.
In fact, when Vedral et al \cite{Vedral1,Vedral2} introduced relative entropy of entanglement,
they pointed out the relation between the hypothesis testing and the relative entropy.
They employed the minimum relative entropy between the given entangled state 
and a separable state.
However, this measure expresses 
the performance of the discrimination between 
the given entangled state and the closest separable state in the sense of relative entropy.
When one hypothesis is chosen to be the set of all separable states,
this measure does not express the performance of this discrimination \cite{Berta}.
Until now, no study discussed the optimal detection efficiency 
for entanglement verification.

In this paper, to resolve this problem, 
we propose a new entanglement measure that expresses 
the optimal performance of the hypothesis testing, i.e., discrimination between 
the given entangled state and the set of separable states.
The reason is that 
quantum hypothesis testing is considered as a key topic in quantum Shannon theory
because it is related to various topics in quantum Shannon theory \cite{Q-cq-channel,H-text,WR,TH13-2,DMHB,Vazquez-Vilar,AJW,CH,BBGL}.
In this characterization, 
in the first step,
we employ the quantum Sanov theorem \cite{Bjelakovic,Notzel}
for hypothesis testing,
and consider the same minimization as 
relative entropy of entanglement by exchanging 
the first and second inputs in the relative entropy.
The proposed quantity is an entanglement monotone.
To enhance this characterization, we derive a much simpler upper bound of 
the error probability in the finite-length setting of this kind of hypothesis testing,
which can be considered another proof of the quantum Sanov theorem.
In addition, we derive the strong converse exponent of the same problem setting as the quantum Sanov theorem.
However, when we simply apply the quantum Sanov theorem,
we assume that 
the separable states are limited to the independent and identical distributed (iid) states of separable states.
To remove this assumption, 
in the second step,
we derive the optimal performance of this discrimination 
when possible separable states are allowed to 
take the form of a general separable state.

This entanglement measure works for finding an
entanglement witness to detect the target entangled state.
This relation is characterized by information geometry \cite{Amari-Nagaoka}.
Also,
we calculate the proposed entanglement measures
when a state is a maximally correlated state.
This class contains all pure states.
For this calculation, we introduce a notable relation with 
coherence measures,
and several calculation formulas for 
coherence measures.

In the next step, we propose an efficient algorithm to compute the proposed entanglement measure on the joint system $\cH_A\otimes \cH_B$.
It is quite difficult to directly calculate entanglement measures
because the calculations of entanglement measures
are directly related to the membership problem of separability, which is known as
NP-hard \cite{Gurvits1,Gurvits2,Gharibian,Ioannou}.
Similar to the paper \cite[Section 3]{Ioannou},
we employ $\delta$-nets on the second system $\cH_B$.
In the following, we assume that the first system $\cH_A$
has a larger or equal dimension.
That is, to optimize the choice of a separable state,
we fix as many states on $\cH_B$ as possible,
and optimize the choice of states on $\cH_A$.
Then, we define an approximate value of the proposed measure. 
Employing the quantum Arimoto-Blahut algorithm proposed by 
the paper \cite{RISB}, 
we propose an algorithm for the approximate value.
This algorithm is an iterative algorithm, and we show 
its convergence to the global minimum and its convergence speed.

As the third step, we apply our algorithm to the membership problem of separability.
Although it is known that this problem is NP-hard \cite{Gurvits1,Gurvits2,Gharibian,Ioannou},
several algorithms have been proposed for this problem.
As a typical method, Doherty et al. \cite{DPS1,DPS2} proposed a method to identify 
whether a given state is a separable state or not
by considering ``symmetric extensions'' and ``PPT-symmetric extensions''
in the second system $\cH_B$.
This kind of extension is called DPS hierarchy.
Later,  Navascu\'{e}s et al. \cite{MOP} evaluated its calculation complexity.
DPS hierarchy was discussed in \cite[(14)]{BCY}, \cite{HNW,Fawzi}.
The paper \cite{Ioannou} proposed another algorithm for this problem by using 
$\delta$-nets on the second system $\cH_B$.
The paper \cite{Ioannou} also discussed its calculation complexity.
It was shown that the calculation complexity of these preceding algorithms
are exponential only for the dimension of the second system $\cH_B$.
The papers \cite{MOP,Ioannou} evaluated the exponential component for 
the dimension of the second system $\cH_B$ in their algorithms.
The paper \cite[Eq. (38)]{MOP} evaluated the polynomial component for 
the dimension of the first system $\cH_A$.
This paper discussed our calculation complexity,
and showed that the calculation complexity of our algorithm
is exponential only for the dimension of the second system $\cH_B$.
In addition, we evaluated the exponential component for 
the dimension of the second system $\cH_B$ in our algorithm,
and evaluated the polynomial component for 
the dimension of the first system $\cH_A$.
As the comparison with the evaluation by \cite{MOP},
we find that our exponential component for 
the dimension of the second system $\cH_B$ 
is larger than that in the analysis \cite{MOP},
but our polynomial component for 
the dimension of the first system $\cH_A$
is smaller than that in the analysis \cite{MOP}.
This fact shows that 
our algorithm has an advantage over the algorithm based on 
``symmetric extensions'' and ``PPT-symmetric extensions''
when the dimension of the second system $\cH_B$ is fixed
and only the dimension of the first system $\cH_A$ increases.

Finally, we apply our algorithm to finding a separation hyperplane.
When we solve a conic linear programming with a general cone,
it is a key step to find a separation hyperplane for an element that does not belong to the cone \cite{GLS,Goberna,Kalantari,Nueda}.
Recently, a conic linear programming with the separable cone
has been applied to several problems in quantum information \cite{H-O,Chitambar}.
Therefore, it is desired to find a separation hyperplane for the separable cone.
In quantum information, a separation hyperplane for the separable cone 
is called an entanglement witness \cite{Bruss,Horodecki,Terhal,Lewenstein,Terhal2}.
From the viewpoint of entanglement theory, 
it is strongly desired to establish a method to find an entanglement witness
for a given state, but such a method has been established only for 
the limited cases \cite{TWKK,SCC}.
In this paper, based on our algorithm, we also propose an algorithm to find 
an entanglement witness, i.e., a separation hyperplane for an entangled state.
Although
the method based on DPS hierarchy does not directly yield an entanglement witness,
our method directly directly yields it, which is an advantage of our method.

Overall, in addition to the proposal of the new entangled measures,
the contributions of this paper are summarized as follows.
\begin{description}
\item[(A1)]
{\bf [Statistical hypothesis testing]}
We have shown that 
the proposed measures present the performance of physical detection
in the sense of hypothesis testing, which has not been pointed out before. 
(Sections \ref{S2} and \ref{S3})
\item[(A2)]
{\bf [Exponential quantum Sanov theorem]}
To support the above fact, we have derived 
the exponential version of the quantum Sanov theorem \cite{Bjelakovic,Notzel}.
This extended version clarifies the relation between
the hypothesis testing with composite hypothesis
and the minimum sandwiched relative R\'{e}nyi entropy,
which was not discussed in the existing studies
\cite{Bjelakovic,Notzel}.
Thanks to this improved evaluation,
we could derive
the performance of the hypothesis testing
when $H_0$ hypothesis is the set of all separable states over $n$ bipartite systems, as opposed to an IID $n$-fold tensor product of a separable state.
(Sections \ref{ESanov} and \ref{S3})

\item[(A3)]
{\bf [Maximally correlated states]}
We have calculated the proposed measures
when the entangled state is maximally correlated.
For this calculation, we have calculated the similar value
for the coherence.
(Sections \ref{S4} and \ref{S5})

\item[(A4)]
{\bf [Algorithm \& membership problem]}
We have shown that 
the proposed measures 
can be calculated by using the generalized 
Arimoto-Blahut algorithm
due to the good shape of the proposed measure, 
which has not been pointed out before.
This fact shows that 
the proposed measures 
can be used for solving the membership problem for separable states.
The superiority between this method and the existing method \cite{DPS1,DPS2,MOP} for the membership problem
depends on the parameters of this problem.
In addition, we have evaluated the relation between 
the calculation complexity and
the precision of one of the proposed algorithms.
We have made a numerical demonstration.
(Sections \ref{S6}, \ref{S7}, and \ref{S8})

\item[(A5)]
{\bf [Finding entanglement witness]}
We have revealed the geometrical structure 
of the proposed measure
by using the viewpoint of information geometry.
Using this geometrical structure, 
we have proposed a method to find an entanglement witness by using the proposed measure.
This fact shows that the calculation of 
the proposed measure can be used for finding 
an entanglement witness of a given state.
Applying the algorithm stated in (A4),
we have proposed an algorithm to find 
an entanglement witness and
such a method has been strongly desired 
from the viewpoint of entanglement theory \cite{TWKK,SCC}.
In addition, information-geometric characterization enables us to simplify the proof of the additivity of the proposed measure while the additivity was shown by \cite{EAP}.
(Sections \ref{S9} and \ref{S7})
\end{description}

The remainder of this paper is organized as follows.
Section \ref{S2} 
formulates the detection of an entangled state 
as a problem of hypothesis testing under the iid setting.
Section \ref{S3}
addresses this problem under the non-iid setting.
Section \ref{S9} explains 
information geometrical characterization of our proposed measure and its relation to an entanglement witness.
Section \ref{S4}
proposes and studies new coherence measures as the preparation of the calculation of our entanglement measures.
Section \ref{S5}
derives several calculation formulas for 
our entanglement measures for maximally entangled states
based on the analysis in Section \ref{S4}.
Section \ref{S6}
proposes algorithms for our entanglement measures
by using $\delta$-net.
Section \ref{S7}
applies our results to the membership problem for separability
and separation hyperplane.
Section \ref{S8}
numerically calculates our entanglement measure,
and check how our method works for the membership problem.
Section \ref{S9} explains that our method finds an entanglement witness. 
Section \ref{ESanov} presents our exponential version of 
quantum Sanov theorem.
Section \ref{S11} makes conclusions and discussions.
Appendix \ref{AP1} presents concrete constructions of $\delta$-nets.
Appendix \ref{AP2} proves the theorem for the relation 
between our entanglement measures
and our proposed coherence measure.
Appendix \ref{AP3} proves theorems for the convergence 
for our algorithms.
Appendix \ref{AP4} proves our exponential version of 
quantum Sanov theorem.
Appendix \ref{AP5} proves our exponential evaluation of 
error probabilities in our problem of hypothesis testing under the non-iid setting.

\section{Detecting entangled state with iid setting}\Label{S2}
We consider the joint system 
$\cH_A\otimes \cH_B$, and discuss how efficiently 
an entangled state $\rho_{AB}$ on $\cH_A\otimes \cH_B$
is distinguished from 
the set ${\cal S}_{AB}$ of separable states on $\cH_A\otimes \cH_B$.
This problem can be formulated as a hypothesis testing problem 
to distinguish $\rho_{AB}$ from ${\cal S}_{AB}$.
For this detection, 
we assume the $n$-fold independently and identical distributed (iid) condition.
That is, we have the following two hypotheses for the possible state as follows.
\begin{align}
H_0:& \hbox{The state is }\rho_{AB}^{\otimes n}.\Label{H0}\\
H_1:& \hbox{The state is }\sigma_{AB}^{\otimes n} \hbox{ with }\sigma_{AB} \in {\cal S}_{AB}\Label{H1}.
\end{align}
To address this problem, we focus on the following quantity.
\begin{align}
&\beta_{\epsilon,n}(\rho_{AB})\nonumber \\
:=&
\min_{0\le T\le I}
\Big\{\Tr T \rho_{AB}^{\otimes n} \Big|
\max_{\sigma_{AB}\in {\cal S}_{AB}}
\Tr (I-T)\sigma_{AB}^{\otimes n}
\le \epsilon
\Big\}.\Label{BN6}
\end{align}

In the above setting, 
we do not impose the LOCC restriction to our measurement unlike \cite{Guhne,Terhal}.
However, these existing studies did not discuss 
the optimal detection efficiency 
for the entanglement verification.
Hence, the above quantity expresses the performance of 
discrimination at the fundamental level.

Here, it is better to explain the relation with 
the entanglement verification \cite{HST,HMT,H-ent,HM,TM,ZH}.
Entanglement verification is the task of identifying a given entangled state within a certain error.
Since entanglement verification is a more difficult task,
the detection efficiency of the entanglement verification
is worse than 
the detection efficiency of the entanglement detection.
Existing studies for entanglement verification
are limited to the verification for 
several special entangled states.
No general theory covers 
the entanglement verification for general entangled states.

The application of quantum Sanov theorem \cite{Bjelakovic,Notzel} guarantees the following
\begin{align}
&\lim_{n \to \infty}-\frac{1}{n} \log \beta_{\epsilon,n}(\rho_{AB})
=E^*(\rho_{AB}) \nonumber \\
:=&\min_{\sigma_{AB} \in {\cal S}_{AB}}
D(\sigma_{AB}\|\rho_{AB}),\Label{BVD}
\end{align}
where 
$D(\sigma\|\rho):= \Tr \sigma (\log \sigma -\log \rho)$.
This value is different from relative entropy of entanglement
$E(\rho_{AB}):= \min_{\sigma_{AB} \in {\cal S}_{AB}} D(\rho_{AB}\| \sigma_{AB})$.
This fact shows that 
our measure $E^*(\rho_{AB})$ expresses 
the limit of the performance to distinguish the entangled state 
$\rho_{AB}$ from ${\cal S}_{AB}$.
In contrast, when the choice of $H_0$ and $H_1$ is exchanged,
$E(\rho_{AB})$ does not have the same meaning 
\cite{Berta}.
The proof of \eqref{BVD} is given as follows.
The part $\le$ in \eqref{BVD} follows from the strong converse part of 
quantum Stein's lemma.
The paper \cite{Bjelakovic} showed the part $\ge$ in \eqref{BVD}.
However, they did not consider the case when $\epsilon$ is exponentially small.
To enhance the meaning of the quantity 
$E^*(\rho_{AB})$,
it is needed to evaluate $\beta_{\epsilon,n}(\rho_{AB})$ in such a case.

\begin{table}[t]
\caption{List of mathematical symbols 
for detecting entangled state
from Section II onwards.}
\label{symbols}
\begin{center}
\begin{tabular}{|l|l|l|}
\hline
Symbol& Description & Eq. number  \\
\hline
\multirow{2}{*}{${\cal S}_{AB}$} &Set of separable states on  &
\\
&$\cH_A\otimes \cH_B$&\\
\hline
\multirow{3}{*}{$\beta_{\epsilon,n}(\rho_{AB})$} & 
Minimum error probability   &  \multirow{3}{*}{\eqref{BN6}}\\
&with constraint $\epsilon$ 
when 
&\\
&$H_1$ is given as \eqref{H1} &\\
\hline
$E^*(\rho_{AB})$ &Minimum relative entropy & \eqref{BVD}\\
\hline
\multirow{2}{*}{$E_\alpha^*(\rho_{AB})$} 
&Minimum relative 
R\'{e}nyi & \multirow{2}{*}{\eqref{BNT}}\\
&entropy of order $\alpha$ & \\
\hline
\multirow{2}{*}{${\cal S}_{AB}^n$} & 
Set of separable states on &\\
& $\cH_A^{\otimes n} \otimes \cH_B^{\otimes n}$ &\\
\hline
\multirow{3}{*}{$\overline{\beta}_{\epsilon,n}(\rho_{AB})$} & 
Minimum error probability   &  \multirow{3}{*}{\eqref{BN6}}\\
&with constraint $\epsilon$ 
when 
&\\
&$H_1$ is given as \eqref{BVR} &\\
\hline
\multirow{2}{*}{$\overline{E}^*(\rho_{AB})$}
&Regularized minimum  &\multirow{2}{*}{\eqref{BVD2}} \\
&relative entropy& \\
\hline
\multirow{2}{*}{$\overline{E}_\alpha^*(\rho_{AB})$}
&Regularized minimum relative 
&\multirow{2}{*}{\eqref{BVD5}} \\
&R\'{e}nyi entropy of order $\alpha$ &\\
\hline
\multirow{2}{*}{$\tilde{E}^*(\rho_{AB})$}
&The above value with 
&\multirow{2}{*}{\eqref{BNF}} \\
&$\alpha=1-0$& \\
\hline
$\sigma_{AB}^*$ &Closest separable state & \eqref{NMP8} \\
\hline
${\cal S}_D$ &Divergence sphere & \eqref{VCE}
\\
\hline
${\cal M}( \sigma_{AB}^*,\rho_{AB})$
& Separation hyperplane & \eqref{BCP}
\\
\hline
\multirow{2}{*}{$|\Phi\rangle$}&
Maximally entangled state
&\\
&$\frac{1}{\sqrt{d}}\sum_{j=0}^{d-1}|j\rangle |j\rangle$ & \\
\hline
\multirow{2}{*}{$E_{{\cal D}}^*(\rho_{AB})$}
&Minimum relative entropy & 
\multirow{2}{*}{\eqref{NMP4}} \\
&with fixed finite set ${\cal D}$& \\
\hline
$\epsilon_2({\cal D})$ 
&& \eqref{VCY}\\
\hline
\multirow{3}{*}{$E_{{\alpha,\cal D}}^*(\rho_{AB})$}
&Minimum relative R\'{e}nyi & 
\multirow{3}{*}{\eqref{NMT4}} \\
&entropy of order $\alpha$ & \\
&with fixed finite set ${\cal D}$& \\
\hline
\end{tabular}
\end{center}
\end{table}

For this aim, using the sandwiched relative R\'{e}nyi entropy
$D_{\alpha}(\sigma\|\rho):=\frac{1}{\alpha-1}\log 
\Tr (\rho^{\frac{1-\alpha}{2\alpha}}
\sigma \rho^{\frac{1-\alpha}{2\alpha}})^\alpha$,
we generalize $E^*(\rho_{AB})$ as
\begin{align}
E_{\alpha}^*(\rho_{AB})
:=\min_{\sigma_{AB} \in {\cal S}_{AB}}D_{\alpha}(\sigma_{AB}\|\rho_{AB}).\Label{BNT}
\end{align}
As presented in Section \ref{ESanov}, we have the following 
strengthened version of \eqref{BVD}.
\begin{theorem}\Label{TH1A}
We have the following relations
\begin{align}
&\lim_{n \to \infty}
-\frac{1}{n}\log \beta_{e^{-nr},n}(\rho_{AB})\nonumber \\
\ge &
\sup_{0\le \alpha < 1}\frac{
(1-\alpha) E_{\alpha}^*(\rho_{AB}) - \alpha r}{1-\alpha}\Label{NBT} ,\\
&-\frac{1}{n} \log \beta_{1-e^{-nr},n}(\rho_{AB})\nonumber \\
\le &
\inf_{\alpha \ge 1}\frac{
(1-\alpha) E_{\alpha}^*(\rho_{AB}) - \alpha r}{1-\alpha}\Label{NBT2}.
\end{align}
More precisely, we have
\begin{align}
&-\frac{1}{n}\log \beta_{e^{-nr},n}(\rho_{AB})\nonumber \\
\ge & 
\sup_{0 \le \alpha < 1}\frac{(1-\alpha) E_{\alpha}^* (\rho_{AB}) - \alpha r}{1-\alpha}\nonumber \\
&-\frac{(d_{AB}+2)(d_{AB}-1)}{n}\log (n+1),
\end{align}
where $d_{AB}$ is the dimension of $\cH_A\otimes \cH_B$.
\end{theorem}
Inequality \eqref{NBT} presents
an attainable exponential rate of the failure probability for detecting the entangled state $\rho_{AB}$
when we impose an exponential constraint for the incorrect detection of any separable state
under the iid condition.
Under the limit $r\to 0$,
the RHS of \eqref{NBT} goes to 
$E_{1}^*(\rho_{AB}):=E^*(\rho_{AB})$.
Hence, the relation \eqref{NBT} recovers the inequality $\ge$ in \eqref{BVD}.
Inequality \eqref{NBT2} presents an upper bound for 
the exponential rate of the failure probability to detect the entangled state $\rho_{AB}$
when we relax the constraints for the incorrect detection of any separable state
under the iid condition.
Due to this inequality, when 
the entangled state $\rho_{AB}$ can detected with a failure probability 
of a strictly larger exponential rate than 
$e^{-n E^*(\rho_{AB})}$,
a separable state is incorrectly detected as an entangled state with a high probability.

When a TP-CP map $\Lambda$ preserves the set ${\cal S}_{AB}$, 
these quantities satisfy the monotonicity as
\begin{align}
E^*(\rho_{AB})
\ge E^*(\Lambda (\rho_{AB})),\quad
E_{\alpha}^*(\rho_{AB})
\ge E_{\alpha}^*(\Lambda (\rho_{AB})).
\end{align}

Our new measure 
$E^*(\rho_{AB})$ has a completely different behavior from 
conventional measures of entanglement.
Most of conventional measures of entanglement of a pure state $\rho_{AB}$
take the von Neumann entropy of its reduced density. 
However, 
$E^*(\rho_{AB})$ becomes infinity 
when $\rho_{AB}$ is a pure entangled state
because the support of no separable state is included in 
the support of $\rho_{AB}$.
Although the log-negativity of a pure state $\rho_{AB}$
does not equal  the von Neumann entropy of its reduced density,
it takes a finite value.

\begin{remark}
When $H_0$ and $H_1$ are exchanged, 
$E(\rho_{AB})$
does not have the same meaning as 
$E^*(\rho_{AB})$.
This is because quantum Sanov theorem works only when 
$H_1$ is composed of several states \cite{Berta}. 
That is, the choice of the hypothesis composed of 
several states is crucial in this scenario.
Instead, to prove that $E(\rho_{AB})$ 
upper bounds the entanglement
distillation rate, 
the paper \cite{Vedral} employs
the simple hypotheses that
$H_0$ is 
$\sigma_{AB}^{\otimes n}$ 
and 
$H_1$ is 
$\rho_{AB}^{\otimes n}$,
where
$\sigma_{AB}$ is an arbitrary chosen element of 
${\cal S}_{AB}$.
More precisely, an application of the strong converse
of Stein's lemma to this hypothesis testing yields
the strong converse for the entanglement distillation rate
\cite[Proof of Theorem 8.7]{HM}.
\end{remark}

\section{Detecting entangled state with non-iid setting}\Label{S3}
The above discussion assumes the iid assumption for the alternative hypothesis.
However, several existing studies for entanglement
verification 
and entanglement detection
assume that 
the alternative hypothesis includes 
general separable states including non-iid separable stets \cite{Saggio,Dimic,
HM,TM,ZH},
which is called the non-iid setting.
Our analysis can be extended to the non-iid setting as follows.
We denote the set of separable states on $\cH_A^{\otimes n} \otimes \cH_B^{\otimes n}$ by ${\cal S}_{AB}^n$. Then, the alternative hypothesis is changed to 
\begin{align}
H_1:& \hbox{The state is }\sigma \hbox{ with }\sigma \in {\cal S}_{AB}^n \Label{BVR}
\end{align}
while the null hypothesis is kept to \eqref{H0}.
To address this problem, we focus on the following quantity.
\begin{align}
&\overline{\beta}_{\epsilon,n}(\rho_{AB})\nonumber \\
:=&
\min_{0\le T\le I}
\Big\{\Tr T \rho_{AB}^{\otimes n} \Big|
\max_{\sigma_{AB}\in {\cal S}_{AB}^n} \Tr (I-T)\sigma_{AB}\le \epsilon
\Big\}.
\end{align}
The relation \eqref{BVD} yields that
\begin{align}
&\lim_{n \to \infty}-\frac{1}{n} \log \overline{\beta}_{\epsilon,n}(\rho_{AB})
\le  \overline{E}^*(\rho_{AB}):=
\inf_{n}\frac{1}{n}{E}^*(\rho_{AB}^{\otimes n})\nonumber \\
=&\lim_{n\to \infty}\frac{1}{n}{E}^*(\rho_{AB}^{\otimes n})
=\lim_{n\to \infty}\frac{1}{n}\min_{\sigma \in {\cal S}_{AB}^n}D(\sigma\|\rho_{AB}^{\otimes n}
).\Label{BVD2}
\end{align}

Extending ${E}_{\alpha}^*(\rho_{AB})$, we 
define 
\begin{align}
\overline{E}_{\alpha}^*(\rho_{AB})
:=&\inf_{n}\frac{1}{n}{E}_\alpha^*(\rho_{AB}^{\otimes n})
=\lim_{n\to \infty}\frac{1}{n}{E}_\alpha^*(\rho_{AB}^{\otimes n})
\nonumber \\
=&\lim_{n\to \infty}\frac{1}{n}\min_{\sigma_{AB} \in {\cal S}_{AB}^n}
D_\alpha(\sigma_{AB}\|\rho_{AB}^{\otimes n}
).\Label{BVD5}
\end{align}
As shown in Appendix \ref{AP5}, we have the following theorem, which is an enhanced version of Theorem \ref{TH1A}, which 
strengthens \eqref{BVD5}.
\begin{theorem}\Label{TH7}
We have
\begin{align}
&\lim_{n \to \infty}
-\frac{1}{n}\log \overline{\beta}_{e^{-nr},n}(\rho_{AB})
\nonumber \\
\ge &
\sup_{0\le \alpha \le 1}\frac{
(1-\alpha) \overline{E}_{\alpha}^*(\rho_{AB}) - \alpha r}{1-\alpha},\Label{NBT3} \\
&\lim_{n \to \infty}-\frac{1}{n} \log \overline{\beta}_{1-e^{-nr},n}(\rho_{AB})\nonumber \\
= &
\inf_{\alpha \ge 1}\frac{
(1-\alpha) \overline{E}_{\alpha}^*(\rho_{AB}) - \alpha r}{1-\alpha}\Label{NBT4}.
\end{align}
More precisely, we have
\begin{align}
&-\frac{1}{n}\log \beta_{e^{-nr},n}(\rho_{AB})\nonumber \\
\ge &
\sup_{0 \le \alpha \le 1}\frac{\frac{1}{n}(1-\alpha) 
E_{\alpha}^* (\rho_{AB}^{\otimes n}) - \alpha r}{1-\alpha}
\nonumber \\
&-\frac{(d_{AB}+2)(d_{AB}-1)}{n}\log (n+1),
\Label{NBT3X}\\
&-\frac{1}{n}\log \beta_{1-e^{-nr},n}(\rho_{AB})\nonumber \\
\le &
\inf_{\alpha \ge 1}\frac{\frac{1}{n}(1-\alpha) 
{E}_{\alpha}^* (\rho_{AB}^{\otimes n}) - \alpha r}{1-\alpha}.
\Label{NBT4X}
\end{align}
\end{theorem}

As an extended version of \eqref{NBT},
Inequality \eqref{NBT3} presents
an attainable exponential rate of the failure probability for detecting the entangled state $\rho_{AB}$
when we impose an exponential constraint for the incorrect detection of any separable state
without the iid condition.
Under the limit $r\to 0$,
the RHS of \eqref{NBT3} goes to 
$\lim_{\delta\to +0} \overline{E}_{1-\delta}^*(\rho_{AB})$.
Inequality \eqref{NBT3} implies the following inequality.
\begin{align}
\lim_{n \to \infty}-\frac{1}{n} \log \overline{\beta}_{\epsilon,n}(\rho_{AB})
\ge \tilde{E}^*(\rho_{AB})
:=\lim_{\delta\to +0} \overline{E}_{1-\delta}^*(\rho_{AB}).
\Label{BNF}
\end{align}
While 
$\overline{E}^*(\rho_{AB})\ge \tilde{E}^*(\rho_{AB})$,
it is not clear whether equality holds in general.
The reference \cite[Proposition 2]{EAP} showed that 
the additivity 
$E^*(\rho_{AB,1}\otimes \rho_{AB,2})
=E^*(\rho_{AB,1})+E^*(\rho_{AB,2})$, which implies
$ {E}^*(\rho_{AB})= \overline{E}^*(\rho_{AB})$
while its simple proof is given in Section \ref{S4-2}.

When the quantity $E_{\alpha}^*(\rho_{AB})$ satisfies the additivity 
\begin{align}
E_{\alpha}^*(\rho_{AB,1}\otimes \rho_{AB,2})
=E_{\alpha}^*(\rho_{AB,1})+E_{\alpha}^*(\rho_{AB,2})
\Label{add-al}
\end{align}
for $\alpha <1$,
the equations
$ {E}^*(\rho_{AB})= \overline{E}^*(\rho_{AB})= \tilde{E}^*(\rho_{AB})$
and 
$ {E}_\alpha^*(\rho_{AB})= \overline{E}_\alpha^*(\rho_{AB})$
hold.
Therefore, it is crucial to clarify when the above additivity holds.
For example, as shown in \eqref{1NMI} of Section \ref{S5},
the maximally correlated state satisfies the above additivity.
The reference \cite[Section 5]{RT24} showed that 
the additivity \eqref{add-al} with $\alpha \in [1/2,1)$
does not hold in general
by considering the bipartite antisymmetric (Werner) state \cite{VW01,ZCH10}\footnote{Since our sandwiched relative entropy of order $\alpha'$
corresponds to $\tilde{D}_{1-\alpha'}$ of \cite[(10)]{RT24}, i.e., the case with $(\alpha,z)=(1-\alpha',\alpha')$.
Hence, the case with $1>\alpha'\ge 1/2$
belongs to the set ${\cal D}$ defined in 
\cite[Section 2]{RT24}.
Hence, the discussion \cite[Section 5]{RT24}
covers the violation of \eqref{add-al}
the additivity \eqref{add-al} with $\alpha \in [1/2,1)$.}.

Inequality \eqref{NBT4} presents an upper bound for 
the exponential rate of the failure probability to detect the entangled state $\rho_{AB}$
when we relax the constraints for the incorrect detection of any separable state
without the iid condition.

\section{Information-geometrical characterization: Finding entanglement witness}\Label{S9}
\subsection{Single system}
Our proposed measure is useful not only for hypothesis testing of the above hypotheses,
but also for finding entanglement witness.
That is, given a non-separable state 
$\rho_{AB} \notin {\cal S}_{AB}$,
it is important to find an entanglement witness
to distinguish $\rho_{AB} $ and $ {\cal S}_{AB}$.
To see this merit,
we discuss its information-geometrical characterization,
which yields entanglement witness.

The papers \cite{TWKK,SCC} proposed methods to find 
an entanglement witness for special states.
As another method, the references \cite{Bertlmann1}, \cite[Section 16.3.3]{Bertlmann2}
proposed a method to use the tangent space 
of the set of separable states based on 
the Hilbert Schmidt norm $\| \cdot \|_{HS}$.
The key point of the above method is 
the fact that the Hilbert-Schmidt norm 
forms the Euclidean geometry over the space of Hermitian matrices.
The above method needs to find 
the closest separable state in the sense of 
the Hilbert Schmidt norm $\| \cdot \|_{HS}$.
However, no method efficiently minimizes 
the Hilbert Schmidt norm $\| \cdot \|_{HS}$
between a given state and the set of separable states,
i.e., efficiently finds 
\begin{align}
\argmin_{ \sigma_{AB} \in {\cal S}_{AB}} 
\| \sigma_{AB}-\rho_{AB}\|_{HS} \Label{HS8}
\end{align}
when the system dimensions are large. 
Hence, it is not easy to find 
an entanglement witness by using the Hilbert-Schmidt norm 
under large Hilbert spaces.

Fortunately, 
this section shows how the minimization 
$\min_{ \sigma_{AB} \in {\cal S}_{AB}} 
D(\sigma_{AB}\|\rho_{AB})$ works 
for finding entanglement witness 
by employing 
the minimizer of $E^*(\rho_{AB})$, i.e., the state
\begin{align}
\sigma_{AB}^*:=
\argmin_{ \sigma_{AB} \in {\cal S}_{AB}} 
D(\sigma_{AB}\|\rho_{AB}). \Label{NMP8}
\end{align}
That is, the minimizer of $E^*(\rho_{AB})$ produces entanglement witness.
In addition, we propose an algorithm to calculate 
our entanglement measure $E^*(\rho_{AB})$
in Section \ref{S6}.
As explained in Section \ref{S7-2}, 
the combination between 
the algorithm for $\sigma_{AB}^*$ and a procedure to 
find an entanglement witness based on quantum relative entropy 
enables us to derive an algorithm to find an entanglement witness.
Hence, the relation between 
entanglement witness and the minimizer of $E^*(\rho_{AB})$ is quite important.

Although quantum relative entropy is a two-input quantity,
it does not form an Euclidean geometry
because it is not symmetric with respect to exchanging 
two input states.
Hence, when we employ our entanglement measure, 
we need to employ a geometric structure different from 
the Euclidean geometry.

Using the operator $(\log \sigma_{AB}^*- \log \rho_{AB})$,
we define the hyperplane
${\cal M}( \sigma_{AB}^*,\rho_{AB})$ as (See Fig. \ref{witness}.)
\begin{align}
&{\cal M}( \sigma_{AB}^*,\rho_{AB})\nonumber \\
:=&
\{\rho'_{AB}| 
\Tr \rho_{AB}' (\log \sigma_{AB}^*- \log \rho_{AB})
= D(\sigma_{AB}^* \|\rho_{AB})
\}.\Label{BCP}
\end{align}

\begin{figure}[htbp]
\begin{center}
\includegraphics[scale=0.6]{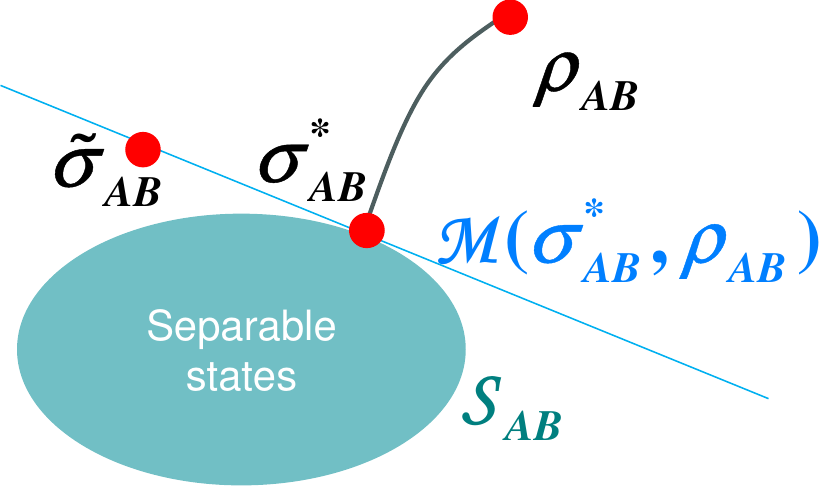}
\end{center}
\caption{Entanglement witness:
This figure shows the relation among
the entangled state $\rho_{AB}$, the set of separable states
${\cal S}_{AB}$, and the separation hyperplane
${\cal M}( \sigma_{AB}^*,\rho_{AB})$.
Any other element $\tilde{\sigma}_{AB}$ on the tangent space
has larger relative entropy to $\rho_{AB}$
than the projected element ${\sigma}_{AB}^*$
due to the Pythagorean theorem \eqref{BVI}.}
\Label{witness}
\end{figure}

The following theorem guarantees that
${\cal M}( \sigma_{AB}^*,\rho_{AB})$
separates $\rho_{AB}$ from ${\cal S}_{AB}$.
Hence, 
the operator $(\log \sigma_{AB}^*- \log \rho_{AB})
-D(\sigma_{AB}^* \|\rho_{AB})$
works as an entanglement witness for $\rho_{AB}$.
In fact, as discussed in \cite[Eq. (4)]{Baichu},
any entanglement witness can be constructed by a combination of local measurement
because any matrix on the composite system 
${\cal H}_A \otimes {\cal H}_B$
can be written as the combination of
at most 
$(\dim {\cal H}_A)^2 (\dim {\cal H}_B)^2$
matrices with a product form.
\begin{theorem}\Label{TH33}
(i) For any separable state $\sigma_{AB}\in {\cal S}_{AB}$,
we have  
\begin{align}
\Tr \sigma_{AB} (\log \sigma_{AB}^*- \log \rho_{AB})
\ge D(\sigma_{AB}^* \|\rho_{AB}).\Label{NM5}
\end{align}
(ii) Also, we have
\begin{align}
\Tr \rho_{AB} (\log \sigma_{AB}^*- \log \rho_{AB})
=- D(\rho_{AB}\| \sigma_{AB}^* )<0.\Label{NM6}
\end{align}
\end{theorem}

To grasp the geometrical meaning of the above fact,
we employ information geometry.
Given a Hermitian matrix $X$ and a state $\tilde{\rho}_0$, 
we often focus on 
the set
${\cal E}:= \{ \tilde{\rho}_t \}$, where
$\tilde{\rho}_t:= \exp(  t X+\log \tilde{\rho}_0)/c_t$
with $c_t:= \Tr \exp(  t X+\log \tilde{\rho}_0)$.
The set ${\cal E}=\{ \tilde{\rho}_t \}$ is called an exponential 
family generated by $X$.
Also, the set
${\cal M}:=\{\rho| \Tr \rho X=\Tr \tilde{\rho}_0 X\}$ is called a mixture family
generated by $X$.

Since the intersection ${\cal M}\cap {\cal E}$ is composed of only one element 
$\tilde{\rho}_0$,
any elements $\sigma \in {\cal M}$ and $\rho \in {\cal E}$
satisfy
\begin{align}
D(\sigma \|\rho)=
D(\sigma \|\tilde{\rho}_0)+ D(\tilde{\rho}_0 \|\rho),\Label{BVI}
\end{align}
which is called Pythagorean theorem \cite{Amari-Nagaoka}, \cite[Theorem 2.3]{H-text}. 

We apply \eqref{BVI} to the case with $X=
-(\log \sigma_{AB}^*- \log \rho_{AB})$
and $\tilde{\rho}_0=\sigma_{AB}^*$.
Any element $\sigma_{AB}' \in {\cal M}( \sigma_{AB}^*,\rho_{AB})$ satisfies
\begin{align}
D(\sigma_{AB}'\| \rho_{AB})=
D(\sigma_{AB}'\| \sigma_{AB}^*)
+D(\sigma_{AB}^*\| \rho_{AB}).
\end{align}

\begin{proof}
Since (ii) is trivial, we show only (i) by contradiction.
We assume that (i) does not hold.
Hence, there exists a separable state $\sigma_{AB}\in {\cal S}$
that does not satisfy \eqref{NM5}.

We define the divergence sphere
\begin{align}
{\cal S}_D:= \{\rho_{AB}'| D( \rho_{AB}'\|\rho_{AB})
=D( \sigma_{AB}^*\|\rho_{AB})
\}.\Label{VCE}
\end{align}
Since 
any state $\sigma_{AB}'$
in ${\cal M}( \sigma_{AB}^*,\rho_{AB})
\setminus \{\sigma_{AB}^*\}$
satisfies 
$D(\sigma_{AB}' \|\rho_{AB})
=D(\sigma_{AB}' \|\sigma_{AB}^*)
+D( \sigma_{AB}^*\|\rho_{AB})
> D( \sigma_{AB}^*\|\rho_{AB})$ due to \eqref{BVI}, 
any state $\sigma_{AB}''$ in the sphere ${\cal S}_D$ satisfies the condition 
$\Tr \sigma_{AB}'' (\log \sigma_{AB}^*- \log \rho_{AB})
\le D(\sigma_{AB}^* \|\rho_{AB})$.
Hence,
the hyperplane 
${\cal M}( \sigma_{AB}^*,\rho_{AB})$  is the tangent space of 
the sphere ${\cal S}_D$ at $\sigma_{AB}^*$.

If a separable state $\sigma_{AB}$ satisfies 
\begin{align}
\Tr \sigma_{AB} (\log \sigma_{AB}^*- \log \rho_{AB})
< D(\sigma_{AB}^* \|\rho_{AB}),
\end{align}
there exists $\lambda \in (0,1)$ such that
$\sigma_{AB}':=\lambda \sigma_{AB}+(1-\lambda)\sigma_{AB}^*$
is inside of the sphere ${\cal S}_D$ as Fig. \ref{Dsphere}.
Hence, 
\begin{align}
D(\sigma_{AB}' \| \rho_{AB})
<D(\sigma_{AB}^* \| \rho_{AB}).\Label{MHU}
\end{align}
Since $\lambda \sigma_{AB}+(1-\lambda)\sigma_{AB}^*$
is also a separable state,
the relation \eqref{MHU} contradicts with the choice of 
$\sigma_{AB}^*$.
\end{proof}

\begin{figure}[htbp]
\begin{center}
\includegraphics[scale=0.6]{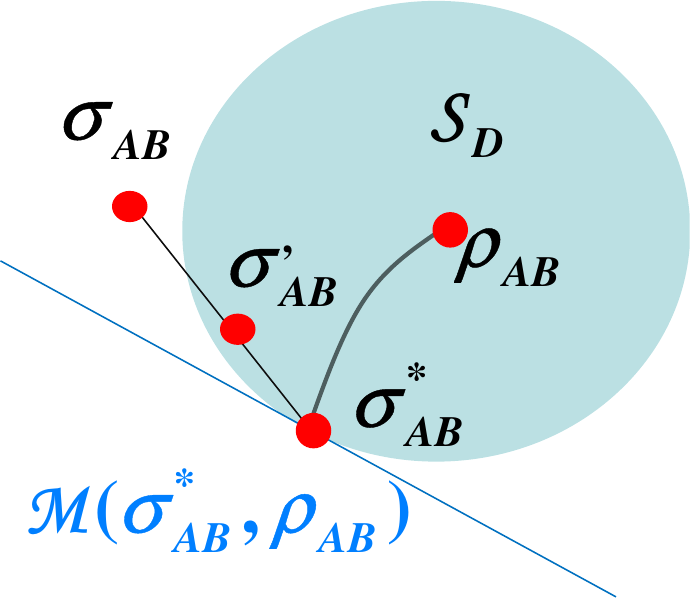}
\end{center}
\caption{Divergence sphere:
This figure shows the relation among
the entangled state $\rho_{AB}$,
the divergence sphere ${\cal S}_D$, and 
the separation hyperplane
${\cal M}( \sigma_{AB}^*,\rho_{AB})$.}
\Label{Dsphere}
\end{figure}

\subsection{Simple proof of additivity}\Label{S4-2}
Next, we consider the composite system of two systems 
${\cal H}_{A,1}\otimes {\cal H}_{B,1}$
and ${\cal H}_{A,2}\otimes {\cal H}_{B,2}$.
We consider the states 
$\rho_1$ and $\rho_2$ on the respective systems.
The following proposition is known.
\begin{proposition}[\protect{\cite[Proposition 2]{EAP}}]\Label{PPH}
\begin{align}
E^*(\rho_1 \otimes \rho_2)=E^*(\rho_1)+ E^*(\rho_2).\Label{BND5}
\end{align}
\end{proposition}
This proposition can be easily shown by using information geometry as follows.

We denote the sets of separable states on 
${\cal H}_{A,j}\otimes {\cal H}_{B,j}$ 
and $({\cal H}_{A,1}\otimes {\cal H}_{A,2})\otimes ({\cal H}_{B,1}\otimes {\cal H}_{B,2})$ 
by ${\cal S}_{AB,j} $ and ${\cal S}_{AB,12} $, respectively, for $j=1,2$.
We choose $\sigma_{j}^*:=\argmin_{ \sigma_{j} \in {\cal S}_{AB,j}} 
D(\sigma_{j}\|\rho_{j})$.
Then, $X_j:=(\log \sigma_{j}^*-\log \rho_j)-D(\sigma_{j}^*\|\rho_j) $
is an entanglement witness for $\rho_j$.
That is, due to Theorem \ref{TH33}, 
any separable state $\sigma_j \in {\cal S}_{AB,j}$ satisfies
\begin{align}
\Tr \sigma_j  X_j \ge 0 \Label{BNE}.
\end{align}

\begin{lemma}
$X_1\otimes I+I\otimes X_2$ is an 
entanglement witness for $\rho_1\otimes \rho_2$.
\end{lemma}
\begin{proof}
For any element $\sigma \in {\cal S}_{AB,1,2}$, we have
\begin{align}
&\Tr \sigma (X_1\otimes I+I\otimes X_2) \notag\\
=& \Tr (\Tr_2\sigma) X_1 + (\Tr_1\sigma) X_2 \ge 0,
\end{align}
where the inequality follows from \eqref{BNE}.
\end{proof}

\begin{proofof}{Proposition \ref{PPH}}
Now, we consider the mixture family
${\cal M}_{12}:=\{\rho| \Tr \rho (X_1\otimes I+I\otimes X_2)=0\}$
generated by the matrix $(X_1\otimes I+I\otimes X_2)$.
We often focus on 
the exponential family 
${\cal E}_{12}:= \{ \rho_{12,t} \}$ generated by $X$, where
$\rho_{12,t}:= \exp(  t (X_1\otimes I+I\otimes X_2) +\log (\rho_1\otimes \rho_2))/c_t$
with $c_{12,t}:= \Tr \exp(  t (X_1\otimes I+I\otimes X_2) +\log (\rho_1\otimes \rho_2))$.

Since the intersection ${\cal M}_{12}\cap {\cal E}_{12}$ is composed of 
only one element $\rho_{12,1}= \sigma_{1}^*\otimes \sigma_{2}^*$,
Eq. \eqref{BVI}
(Pythagorean theorem \cite{Amari-Nagaoka}, \cite[Theorem 2.3]{H-text}) 
guarantees that any element $\sigma \in {\cal M}_{12}$ 
satisfies
\begin{align}
D(\sigma \|\rho_1\otimes \rho_2)=&
D(\sigma \|\sigma_{1}^*\otimes \sigma_{2}^*)
+ D(\sigma_{1}^*\otimes \sigma_{2}^* \|\rho_1\otimes \rho_2)\notag\\
\ge &
D(\sigma_{1}^*\otimes \sigma_{2}^* \|\rho_1\otimes \rho_2).
\Label{BVHG}
\end{align}
Since any separable state $\sigma' \in {\cal S}_{AB,12}$ is located in the opposite side of 
$\rho_1\otimes \rho_2$ with respect to hyperplane ${\cal M}_{12}$, 
the relation \eqref{BVHG} yields
\begin{align}
D(\sigma' \|\rho_1\otimes \rho_2)
\ge 
D(\sigma_{1}^*\otimes \sigma_{2}^* \|\rho_1\otimes \rho_2),
\Label{BVT}
\end{align}
which implies
\begin{align}
\min_{\sigma' \in {\cal S}_{AB,12}}D(\sigma' \|\rho_1\otimes \rho_2)
=D(\sigma_{1}^*\otimes \sigma_{2}^* \|\rho_1\otimes \rho_2).
\end{align}
Therefore, we obtain \eqref{BND5}.
\end{proofof}

\section{Coherence measure}\Label{S4}
Next, we consider how to calculate the proposed measure for maximally correlated states, which are special cases.
For its preparation, 
we study the coherence measures defined below,
because,
as explained in the next section,
our entanglement measures are related to
these coherence measures.
In particular, 
the analysis of maximally correlated states 
can be reduced to the analysis of the coherence measures.
We consider a state 
$\sum_{i,j}\theta_{i,j}|j\rangle \langle i|$ on a
$d$-dimensional  Hilbert space $\cH$ spanned by 
$\{|j\rangle \}_{j=0}^{d-1}$,
and denote the set of incoherent states
by ${\cal I}_{c}$, which is written as $\{\sigma(p)\}$ with $\sigma(p):=\sum_{i}p_i |i\rangle \langle i|$, where $p=(p_i)$
is a distribution.
Then, we define the following coherence measures.
\begin{align}
C^*(\rho):=\min_{\sigma \in {\cal I}_{c}}D(\sigma\|\rho),\quad
C_\alpha^*(\rho):=\min_{\sigma \in {\cal I}_{c}}D_\alpha(\sigma\|\rho).
\Label{ZXO}
\end{align}
Although the papers \cite{Chitambar2,ZHC} considered
the coherence measures based on R\'{e}nyi relative entropy,
our measures are different from their coherence measures
because their measures are based on the quantities 
with exchanging the two input states from ours.
In fact, $C^*(\rho)$ has a geometrical characterization 
similar to Section \ref{S9} by replacing 
${\cal S}_{AB}$ by ${\cal I}_{c}$.

For $\alpha\in (0,1)$ and 
a pure state $|\phi\rangle\langle \phi|$ with
$|\phi\rangle=\sum_{j} e^{\theta_j \sqrt{-1}}p_j|j\rangle$,
we have
\begin{align}
&e^{(\alpha-1) C_{\alpha}^*(|\phi\rangle\langle \phi|)}
=\max_{\sigma \in {\cal I}_{c}}
\Tr (|\phi\rangle\langle \phi|^{\frac{1-\alpha}{2\alpha}}
\sigma |\phi\rangle\langle \phi|^{\frac{1-\alpha}{2\alpha}})^\alpha\nonumber \\
=&\max_{\sigma \in {\cal I}_{c}}
\Tr (|\phi\rangle \langle \phi|
\sigma |\phi\rangle \langle \phi|)^\alpha \nonumber \\
=& \max_{\sigma(q) \in {\cal I}_{c}}
(\langle \phi|
\sigma(q) |\phi\rangle )^\alpha 
= \max_{q}
(\sum_{j}q_j p_j)^\alpha
=(\max_{i} p_i)^\alpha.
\end{align}
Hence, we have
\begin{align}
C_{\alpha}^*(|\phi\rangle\langle \phi|)=\frac{\alpha}{\alpha-1}\log (\max_{i} p_i).
\Label{BMT2}
\end{align}
Also, $C^*(|\phi\rangle\langle \phi|)$ and 
$C_{\alpha}^*(|\phi\rangle\langle \phi|)$ take
infinity value for $\alpha>1$ unless it is an incoherent state.

In the following discussion, we employ the operators
$Z:= \sum_{j=0}^{d-1} \omega^j|j\rangle \langle j|$ with $\omega
:= e^{2 \pi \sqrt{-1}/d}$
and
$X:=\sum_{j=0}^{d-1}|j+1\rangle \langle j| $.
We define
$\rho(p):= \sum_{j=0}^{d-1}p_j Z^j |+\rangle \langle +| Z^{-j}$ with 
$|+\rangle := \frac{1}{\sqrt{d}}\sum_{j=0}^{d-1}|j\rangle$.
Then, we have the following lemma.
\begin{lemma}
Using the uniform distribution $p_{mix}$, we have
\begin{align}
C^*(\rho(p))
&= D( p_{\mix}\|p)=  -\frac{1}{d}\sum_{j=0}^{d-1}\log p_j -\log d \Label{NRUT} \\
C_\alpha^*(\rho(p))
&= D_\alpha( p_{\mix}\|p)\nonumber \\
&= \frac{1}{\alpha-1}\log \sum_{j=0}^{d-1}p_j^{1-\alpha}-
\frac{\alpha}{\alpha-1}\log d,\Label{NRU2T}
\end{align}
for $\alpha \in (0,1)\cup (1,2]$.
\end{lemma}

\begin{proof}
Since the above second equations follow from the direct calculation, the following analysis shows the first equations.

We have
$X \rho(p) X^{-1}=\rho(p).$
Thus, $D( \sigma(q)\|\rho(p))
=D(X^j \sigma(q)X^{-j}\|\rho(p))$. 
Hence, the joint convexity of the relative entropy implies
\begin{align}
&D( \sigma(q)\|\rho(p))
=\sum_{j=0}^{d-1} \frac{1}{d}D(X^j \sigma(q)X^{-j}\|\rho(p))
\nonumber \\
\ge &
D(\sum_{j=0}^{d-1} \frac{1}{d} X^j \sigma(q)X^{-j}\|\rho(p))\nonumber  \\
=& D(\rho_{mix}\|\rho(p))
=D(\rho(p_{mix})\|\rho(p)) =D(p_{mix}\|p),
\end{align}
where $\rho_{mix}$ is the completely mixed state.
The above relations imply the first equation \eqref{NRUT}.
Similarly, since $x \mapsto x^\alpha$ is matrix convex for 
$\alpha\in [1,2]$, 
\begin{align}
& e^{(\alpha-1)D_\alpha( \sigma(q)\|\rho(p))}
=\sum_{j=0}^{d-1} \frac{1}{d}e^{(\alpha-1)D_\alpha( X^j\sigma(q)X^j \|\rho(p))} \nonumber \\
\ge &
e^{(\alpha-1)D_\alpha( \sum_{j=0}^{d-1} \frac{1}{d} X^j\sigma(q)X^j \|\rho(p))} \nonumber \\
=&e^{(\alpha-1)D_\alpha( \rho_{mix} \|\rho(p))}
=e^{(\alpha-1)D_\alpha( p_{mix} \|p)},\Label{NBYA}
\end{align}
which implies the first equation \eqref{NRU2T} for $\alpha\in (1,2]$.
Similarly, since $x \mapsto x^\alpha$ is matrix concave for $\alpha\in [0,1)$, 
an analysis similar to \eqref{NBYA}
yields the first equation \eqref{NRU2T} for $\alpha\in [0,1)$.
\end{proof}

To calculate $C^*(\rho)$ for a general state,
we define the state 
$I_{c}(\rho):=
\frac{1}{t(\rho)} \exp((\log \rho)^{\rm diag})$
with $t(\rho):=\Tr \exp((\log \rho)^{\rm diag})$
and $(\log \rho)^{\rm diag}$ is the diagonal 
matrix whose diagonal elements are the same as
$\log \rho$.
Then, $C^*(\rho)$ is calculated as follows.

\begin{lemma}\Label{L2}
The relation
\begin{align}
C^*(\rho)= 
D(I_{c}(\rho) \|\rho)\Label{BVP}
\end{align}
holds.
\end{lemma}

\begin{proof}
For $\sigma \in {\cal I}_{c}$, we have 
\begin{align}
&D(\sigma\|\rho)
=\Tr \sigma (\log \sigma - \log \rho)\nonumber \\
=&\Tr \sigma (\log \sigma - (\log \rho)^{\rm diag})) \nonumber \\
=&\Tr \sigma (\log \sigma - \log I_{c}(\rho)-\log t(\rho)
\nonumber \\
=&D( \sigma \| I_{c}(\rho))-\log t(\rho).
\end{align}
Hence, the minimization is realized when $\sigma=I_{c}(\rho)$.
\end{proof}

\begin{lemma}\Label{L3}
The relation
\begin{align}
&C_\alpha^*(\rho)
\nonumber \\
=&\left\{
\begin{array}{l}
\max_{\tau} 
\frac{\alpha}{\alpha-1}\log \|( \rho^{\frac{1-\alpha}{2\alpha}}
\tau^{-\frac{1-\alpha}{\alpha}} \rho^{\frac{1-\alpha}{2\alpha}})^{\rm diag}
\| \\
\hspace{35ex} \hbox{ for }  \alpha \in [1/2,1)\\
\max_{\tau} 
\frac{-\alpha}{\alpha-1}\log \|(( \rho^{\frac{1-\alpha}{2\alpha}}
\tau^{-\frac{1-\alpha}{\alpha}} \rho^{\frac{1-\alpha}{2\alpha}})^{\rm diag})^{-1}
\| \\
\hspace{35ex} \hbox{ for }  \alpha >1
\end{array}
\right.
\Label{NYH2}
\end{align}
holds.
\end{lemma}

\begin{proof}
The idea of this proof comes from the proof of \cite[Lemma 6 and (B4)]{HT}. 
We choose 
$\beta$ such that $\frac{1}{\beta}+\frac{1}{\alpha}=1$
so that $\beta= \frac{\alpha-1}{\alpha}$.
We fix $\alpha \in [1/2,1)$ so that $\beta\in [-1,0)$.
For any state $\tau$ and any incoherent state 
$\sigma \in {\cal I}_{c}$,
the reverse matrix H\"older inequality \cite[Theorem 2.7]{H-q-text} implies
\begin{align}
&\Tr 
\tau^{\beta} 
\rho^{\frac{1-\alpha}{2\alpha}}
 \sigma \rho^{\frac{1-\alpha}{2\alpha}}
\ge
(\Tr \tau)^{\beta}
(\Tr 
(\rho^{\frac{1-\alpha}{2\alpha}}
 \sigma \rho^{\frac{1-\alpha}{2\alpha}})^\alpha)^{\frac{1}{\alpha}} \nonumber \\
 =&
(\Tr 
(\rho^{\frac{1-\alpha}{2\alpha}}
 \sigma \rho^{\frac{1-\alpha}{2\alpha}})^\alpha)^{\frac{1}{\alpha}}.
\end{align}
Considering the equality condition, we have
\begin{align}
&\min_{\tau} 
\Tr 
\tau^{\beta} 
\rho^{\frac{1-\alpha}{2\alpha}}
 \sigma \rho^{\frac{1-\alpha}{2\alpha}}
=
(\Tr 
(\rho^{\frac{1-\alpha}{2\alpha}}
 \sigma \rho^{\frac{1-\alpha}{2\alpha}})^\alpha)^{\frac{1}{\alpha}}.
\end{align}
Hence,
\begin{align}
&e^{(\alpha-1)C_\alpha^*(\rho)}
=
\max_{\sigma \in {\cal I}_{c}} 
\Tr 
(\rho^{\frac{1-\alpha}{2\alpha}}
 \sigma \rho^{\frac{1-\alpha}{2\alpha}})^\alpha
\nonumber \\
 =&
(\max_{\sigma \in {\cal I}_{c}} \min_{\tau}
\Tr \rho^{\frac{1-\alpha}{2\alpha}}
\tau^{\beta} \rho^{\frac{1-\alpha}{2\alpha}}
 \sigma )^\alpha \nonumber \\
\stackrel{(a)}{=}&(\min_{\tau}\max_{\sigma \in {\cal I}_{c}} 
\Tr \rho^{\frac{1-\alpha}{2\alpha}}
\tau^{\beta} \rho^{\frac{1-\alpha}{2\alpha}}
 \sigma )^\alpha
\nonumber \\
=&
\min_{\tau} 
\|( \rho^{\frac{1-\alpha}{2\alpha}}
\tau^{\beta} \rho^{\frac{1-\alpha}{2\alpha}})^{\rm diag}
\|^\alpha.
\end{align}
Here,
$(a)$ follows from the minimax theorem because
$\Tr \rho^{\frac{1-\alpha}{2\alpha}}
\tau^{\beta} \rho^{\frac{1-\alpha}{2\alpha}}
 \sigma$ is linear for $\sigma$ and is convex for $\tau$,
 which follows from 
 the matrix convexity of $\tau \mapsto \tau^{\beta}$ with $\beta\in [-1,0)$.
 Hence, we obtain the desired equation with $\alpha \in (0,1)$.

Next, we fix $\alpha >1$ so that $1>\beta>0$.
The matrix H\"older inequality \cite{Hiai},\cite[Theorem 2.6]{H-q-text}
implies
\begin{align}
&\Tr 
\tau^{\beta} 
\rho^{\frac{1-\alpha}{2\alpha}}
 \sigma \rho^{\frac{1-\alpha}{2\alpha}}
\le
(\Tr \tau)^{\beta}
(\Tr 
(\rho^{\frac{1-\alpha}{2\alpha}}
 \sigma \rho^{\frac{1-\alpha}{2\alpha}})^\alpha)^{\frac{1}{\alpha}}\nonumber \\
 =&
(\Tr 
(\rho^{\frac{1-\alpha}{2\alpha}}
 \sigma \rho^{\frac{1-\alpha}{2\alpha}})^\alpha)^{\frac{1}{\alpha}}.
\end{align}
Considering the equality condition, we have
\begin{align}
&\max_{\tau} 
\Tr 
\tau^{\beta} 
\rho^{\frac{1-\alpha}{2\alpha}}
 \sigma \rho^{\frac{1-\alpha}{2\alpha}}
=
(\Tr 
(\rho^{\frac{1-\alpha}{2\alpha}}
 \sigma \rho^{\frac{1-\alpha}{2\alpha}})^\alpha)^{\frac{1}{\alpha}}.
\end{align}
Hence,
\begin{align}
&e^{(\alpha-1)C_\alpha^*(\rho)}
=
\min_{\sigma \in {\cal I}_{c}} 
\Tr 
(\rho^{\frac{1-\alpha}{2\alpha}}
 \sigma \rho^{\frac{1-\alpha}{2\alpha}})^\alpha\nonumber \\
 =&
(\min_{\sigma \in {\cal I}_{c}} \max_{\tau}
\Tr \rho^{\frac{1-\alpha}{2\alpha}}
\tau^{\beta} \rho^{\frac{1-\alpha}{2\alpha}}
 \sigma )^\alpha \nonumber \\
\stackrel{(a)}{=}&(\max_{\tau}\min_{\sigma \in {\cal I}_{c}} 
\Tr \rho^{\frac{1-\alpha}{2\alpha}}
\tau^{\beta} \rho^{\frac{1-\alpha}{2\alpha}}
 \sigma )^\alpha
\nonumber \\
=&
(\max_{\tau} 
\|
(( \rho^{\frac{1-\alpha}{2\alpha}}
\tau^{\beta} \rho^{\frac{1-\alpha}{2\alpha}})^{\rm diag})^{-1}
\|^{-1})^{\alpha}
\nonumber \\
=&
\max_{\tau} 
\|
(( \rho^{\frac{1-\alpha}{2\alpha}}
\tau^{\beta} \rho^{\frac{1-\alpha}{2\alpha}})^{\rm diag})^{-1}
\|^{-\alpha}.
\end{align}
Here,
$(a)$ follows from the minimax theorem because
$\Tr \rho^{\frac{1-\alpha}{2\alpha}}
\tau^{\beta} \rho^{\frac{1-\alpha}{2\alpha}}
 \sigma$ is linear for $\sigma$ and is concave for $\tau$.
Hence, we obtain the desired equation with $\alpha >1$.
\end{proof}

Coherence measures satisfy the additivity.
\begin{lemma}\Label{L4}
The relations
\begin{align}
\begin{split}
C^*(\rho_1\otimes \rho_2)
&= C^*(\rho_1)+ C^*(\rho_2),\\
C_\alpha^*(\rho_1\otimes \rho_2)
&= C_\alpha^*(\rho_1)+ C_\alpha^*(\rho_2) 
\end{split}
\Label{NYH}
\end{align}
hold with $\alpha \in [1/2,1) \cup (1,\infty)$.
\end{lemma}

\begin{proof}
Since the first equation in \eqref{NYH} follows from 
the second equation in \eqref{NYH} with taking the limit $\alpha\to 1$, we show only the second equation in \eqref{NYH}.
The idea of this proof is similar to the proof of \cite[Lemma 7]{HT}.
For $\alpha \in (0,1)$, using \eqref{NYH2}, we have
\begin{align}
&C_\alpha^*(\rho_1)+ C_\alpha^*(\rho_2) 
\ge 
C_\alpha^*(\rho_1\otimes \rho_2) \nonumber \\
=&
\max_{\tau} 
\frac{\alpha}{\alpha-1}\log \big\|( (\rho_1\otimes \rho_2)^{\frac{1-\alpha}{2\alpha}}
\tau^{-\frac{1-\alpha}{\alpha}} (\rho_1\otimes \rho_2)^{\frac{1-\alpha}{2\alpha}})^{\rm diag}
\big\| \nonumber \\
\ge &
\max_{\tau_1,\tau_2} 
\frac{\alpha}{\alpha-1}\log \big\|( (\rho_1\otimes \rho_2)^{\frac{1-\alpha}{2\alpha}}
(\tau_1\otimes \tau_2)^{-\frac{1-\alpha}{\alpha}} 
\nonumber \\
&\hspace{20ex}\cdot(\rho_1\otimes \rho_2)^{\frac{1-\alpha}{2\alpha}})^{\rm diag}
\big\| \nonumber \\
= &
\max_{\tau_1,\tau_2} 
\frac{\alpha}{\alpha-1}
\big(\log \big\|( \rho_1^{\frac{1-\alpha}{2\alpha}}
\tau_1^{-\frac{1-\alpha}{\alpha}} \rho_1^{\frac{1-\alpha}{2\alpha}})^{\rm diag}
\big\| \nonumber \\
&+
\log \big\|( \rho_2^{\frac{1-\alpha}{2\alpha}}
\tau_2^{-\frac{1-\alpha}{\alpha}} \rho_2^{\frac{1-\alpha}{2\alpha}})^{\rm diag}
\big\| \big)
\nonumber  \\
=&
C_\alpha^*(\rho_1)+ C_\alpha^*(\rho_2) ,
\end{align}
which implies the second equation in \eqref{NYH}.
The case with $\alpha >1$ can be shown in the same way.
\end{proof}

\begin{table}[t]
\caption{List of mathematical symbols 
for coherence.}
\label{symbols2}
\begin{center}
\begin{tabular}{|l|l|l|}
\hline
Symbol& Description & Eq. number  \\
\hline
${\cal I}_{c}$ &Set of incoherent states &
\\
\hline
$C^*(\rho)$ &Minimum relative entropy & \eqref{ZXO}\\
\hline
\multirow{2}{*}{$C_\alpha^*(\rho)$} 
&Minimum relative 
R\'{e}nyi & \multirow{2}{*}{\eqref{ZXO}}\\
&entropy of order $\alpha$ & \\
\hline
$I_{c}(\rho)$ & Closest incoherent state & 
\eqref{BVP}\\
\hline
$\sigma(p)$ & $\sum_{i}p_i |i\rangle \langle i|$&
\\
\hline
$\rho(p)$ & $ \sum_{j=0}^{d-1}p_j Z^j |+\rangle \langle +| Z^{-j}$&
\\
\hline
\end{tabular}
\end{center}
\end{table}

\begin{remark}
Although the above-presented proofs employ
several techniques that come from prior works \cite{Chitambar2,ZHC},
the above proofs are not simple combinations of 
techniques that come from prior works.
Since the two entries of the relative entropy
and relative R\'{e}nyi entropy  
are exchanged in our measure,
the above proofs require new
careful treatments for matrix convexity
related to the parameter $\alpha$.
\end{remark}

\begin{remark}
It is possible to make the same discussion for 
the coherence as Sections \ref{S2} -- \ref{S9}.
Since the coherence measures satisfy the additivity as Lemma \ref{L4},
the detectability even with the non-iid setting like Section \ref{S3}
can be characterized by the quantities 
$C^*(\rho)$ and $C_\alpha^*(\rho)$.
\end{remark}

\section{Maximally correlated states}\Label{S5}
The discussion in Section \ref{S2}
shows
the operational meaning of 
$E_{\alpha}^*(\rho_{AB})$ including $E^*(\rho_{AB})$.
We calculate 
$E^*(\rho_{AB})$ and
$E^*_\alpha(\rho_{AB})$ for several special cases.
When $\rho_{AB}$ is a pure state 
$|\phi\rangle \langle \phi|$
$E^*(\rho_{AB})$ and
$E^*_\alpha(\rho_{AB})$ with $\alpha>1$
are infinity.
However,
$E^*_\alpha(\rho_{AB})$ with $\alpha \in [0,1)$
takes a finite value as follows.
When the Schmidt coefficients of 
$|\phi\rangle \langle \phi|$ are $\sqrt{p_i}$ with $i=1, \ldots,d$,
$\max_{\sigma_{AB} \in {\cal S}_{AB}} \langle \phi| \sigma |\phi\rangle=
\max_{i} p_i$ \cite[(40) and (97)]{ZHC}.
Hence,
\begin{align}
&e^{(\alpha-1) E_{\alpha}^*(\rho_{AB})}
=\max_{\sigma_{AB} \in {\cal S}_{AB}}
\Tr \big(\rho^{\frac{1-\alpha}{2\alpha}}
\sigma_{AB} \rho^{\frac{1-\alpha}{2\alpha}}\big)^\alpha
\nonumber \\
=&\max_{\sigma_{AB} \in {\cal S}_{AB}}
\Tr \big(|\phi\rangle \langle \phi|
\sigma_{AB} |\phi\rangle \langle \phi|\big)^\alpha 
= \max_{\sigma_{AB} \in {\cal S}_{AB}}
\langle \phi| \sigma_{AB} |\phi\rangle^\alpha
\nonumber \\
=&\Big(\max_{\sigma_{AB} \in {\cal S}_{AB}} \langle \phi| \sigma_{AB} |\phi\rangle \Big)^\alpha
=(\max_{i} p_i)^\alpha.
\end{align}
Hence, we have
\begin{align}
E_{\alpha}^*(\rho_{AB})=\frac{\alpha}{\alpha-1}\log (\max_{i} p_i).
\Label{BMT}
\end{align}

More generally, the calculation of these quantities 
can be converted to the calculation of the coherence measure
for special states $\rho_{AB}$.
Similar to the paper \cite{ZHC}, we consider maximally correlated states.
A state $\rho_{AB}$ on $\cH_A \otimes \cH_B$
is called a maximally correlated state when
there exist bases $\{|u_{j,A}\rangle\}_j$ and $\{|u_{j,B}\rangle\}_j$
of $\cH_A$ and $\cH_B$ such that
the support $\rho_{AB}$ is included in the subspace spanned by 
$\{|u_{j,A}\rangle|u_{j,B}\rangle\}_j$ on $\cH_A\otimes \cH_B$ 
\cite{Rains,Hiroshima}, \cite[Section 8.8]{H-text}. 
In this case, the state $\rho_{AB}$
is written as 
$\sum_{i,j}\theta_{i,j}
|u_{j,A}\rangle|u_{j,B}\rangle \langle u_{i,A}|\langle u_{i,B}|$.
In this case, the calculation of 
$E^*(\rho_{AB})$ and
$E^*_\alpha(\rho_{AB})$ can be simplified as follows.
Similar to Theorem 1 in \cite{ZHC}, 
the calculation of 
$E^*(\rho_{AB})$ and
$E^*_\alpha(\rho_{AB})$
are reduced to coherence measures.

\begin{lemma}\Label{L1}
The relations
\begin{align}
E^* \Big(\sum_{i,j}\theta_{i,j}
|u_{j,A}\rangle|u_{j,B}\rangle \langle u_{i,A}|\langle u_{i,B}|\Big)
=& C^*\Big(\sum_{i,j}\theta_{i,j}|j\rangle \langle i|\Big) \Label{L1A}\\
E_\alpha^*\Big(\sum_{i,j}\theta_{i,j}
|u_{j,A}\rangle|u_{j,B}\rangle \langle u_{i,A}|\langle u_{i,B}|\Big)
=& C_\alpha^*\Big(\sum_{i,j}\theta_{i,j}|j\rangle \langle i|\Big) 
\Label{L1B}
\end{align}
hold.
\end{lemma}
Lemma \ref{L1} is shown in Appendix \ref{AP2}.
Since a pure state is a maximally correlated state,
the calculation \eqref{BMT} can be recovered by 
Lemma \ref{L1} and \eqref{BMT2}.
Also, $E^* \Big(\sum_{i,j}\theta_{i,j}
|u_{j,A}\rangle|u_{j,B}\rangle \langle u_{i,A}|\langle u_{i,B}|\Big)
$ can be calculated by Lemma \ref{L2} as
\begin{align}
E^* \Big(\sum_{i,j}\theta_{i,j}
|u_{j,A}\rangle|u_{j,B}\rangle \langle u_{i,A}|\langle u_{i,B}|\Big)
=D(I_{c}(\rho_\theta)\| \rho_\theta),
\end{align}
where $\rho_\theta:=\sum_{i,j}\theta_{i,j}|j\rangle \langle i|$.
The combination of Lemmas \ref{L4} and \ref{L1} guarantees 
that any maximally entangled state $\rho_{AB}$ satisfies
\begin{align}
E^*(\rho_{AB})= 
\overline{E}^*(\rho_{AB})=
\tilde{E}^*(\rho_{AB}), ~
E_\alpha^*(\rho_{AB})= 
\overline{E}_\alpha^*(\rho_{AB}).
\Label{1NMI}
\end{align}
The iid setting and the non-iid setting have the same bounds.

To find another example, 
we set $\cH_A$ and $\cH_B$ to be $d$-dimensional spaces,
and define the maximally entangled state
$|\Phi\rangle := \frac{1}{\sqrt{d}}\sum_{j=0}^{d-1}|j\rangle |j\rangle$.
We define the state
$\rho_{AB}(p):= \sum_{j=0}^{d-1}p_j Z^j \otimes I|\Phi\rangle \langle \Phi| Z^{-j}\otimes I$.
The combination of \eqref{NRUT}, \eqref{NRU2T},
Lemma \ref{L1} implies 
\begin{align}
E^*(\rho_{AB}(p))
&= D( p_{\mix}\|p)=  -\frac{1}{d}\sum_{j=0}^d\log p_j -\log d \Label{NRU} \\
E_\alpha^*(\rho_{AB}(p))
&= D_\alpha( p_{\mix}\|p)\nonumber \\
&= \frac{1}{\alpha-1}\log \sum_{j=0}^{d-1}p_j^{1-\alpha}-
\frac{\alpha}{\alpha-1}\log d\Label{NRU2}
\end{align}
for $\alpha \in (0,1)\cup (1,2]$.
However,
it is not easy to calculate these quantities for general states.
Therefore, it is needed to develop an efficient method for their calculation in the general case.

\section{Algorithms}\Label{S6}
\subsection{Algorithm for $E^*(\rho_{AB})$}
Although the previous section presents the exact values
of the proposed measures of maximally correlated states,
it is not so easy to calculate the proposed measures of a general density matrix.
To calculate the proposed quantities for general states, 
we introduce algorithms for their calculation that can be applied to 
general states.
Since the calculation of $E^*(\rho_{AB})$ is hard,
we modify $E^*(\rho_{AB})$ as follows.
We choose a finite set ${\cal D}$ of pure states on ${\cal H}_B$
as a discretization.
That is, we denote ${\cal D}$
as $\{ \rho_{B|c}^*\}_{c \in {\cal X}}$.
Then, we define
\begin{align}
E_{{\cal D}}^*(\rho_{AB}):=
\min_{ P(c),\rho_{A|c}} 
D\Big( \sum_{c}  P(c) \rho_{A|c}\otimes \rho_{B|c}^*
\Big\|\rho_{AB}\Big) \Label{NMP4}.
\end{align}

We define $\epsilon_2({\cal D})$ as
\begin{align}
\epsilon_2({\cal D}):=
2\sqrt{1- 
\min_{\psi_1} \max_{\psi_2\in {\cal D}}
|\langle \psi_1|\psi_2\rangle|^2}.\Label{VCY}
\end{align}
We have
\begin{align}
& \max_{P'(c), \rho_{A|c}', \rho_{B|c}'}
\min_{P(c), \rho_{A|c}}
\Big\|\sum_{c}  P(c) \rho_{A|c}\otimes \rho_{B|c}^*
\nonumber \\
&-\sum_{c}  P'(c) \rho_{A|c}'\otimes \rho_{B|c}' \Big\|_1 \nonumber \\
\le &
\max_{\psi_1} \min_{\psi_2\in {\cal D}}
\| 
 |\psi_1 \rangle \langle \psi_1|
- |\psi_2 \rangle \langle \psi_2| \|_1\nonumber \\
=&2\sqrt{1-\min_{\psi_1} \max_{\psi_2\in {\cal D}}
|\langle \psi_1|\psi_2\rangle|^2}=\epsilon_2({\cal D}).\Label{NNU}
\end{align}
Therefore, 
Fannes inequality implies 
\begin{align}
& \max_{P'(c), \rho_{A|c}', \rho_{B|c}'}
\min_{P(c), \rho_{A|c}}
\Big| D\Big( \sum_{x}  P(x) \rho_{A|x}\otimes \rho_{B|x}^*
\Big\|\rho_{AB}\Big)
\nonumber \\
&\hspace{20ex}-
D\Big( \sum_{x}  P'(x) \rho_{A|x}'\otimes \rho_{B|x}'
\Big\|\rho_{AB}\Big)
\Big|\nonumber \\
\le & 
\epsilon_2({\cal D})( \| \log \rho_{AB} \| +\log d_A d_B)
+\eta_0 (\epsilon_2({\cal D})),
\end{align}
where
$\eta_0$ is defined as
\begin{align}
\eta_0(x):=
\left\{
\begin{array}{ll}
-x \log x & \hbox{when } 0 \le x  \le 1/e \\
1/e & \hbox{otherwise.}
\end{array}
\right.
\end{align}
Therefore, we have
\begin{align}
& E_{{\cal D}}^*(\rho_{AB})\nonumber \\
&-\Big(\epsilon_2({\cal D})( \| \log \rho_{AB} \| +\log d_A d_B)
+\eta_0 (\epsilon_2({\cal D})) \Big)\nonumber \\
\le & E_{}^*(\rho_{AB})  \le E_{{\cal D}}^*(\rho_{AB}).\Label{NGR}
\end{align}

We rewrite $P(c)\rho_{A|c} $ to $X_c$.
Using the idea of the algorithm by \cite{RISB}, 
we propose an iterative algorithm to calculate 
$E_{{\cal D}}^*(\rho_{AB})$, in which 
matrices $X_{A|c}$ is updated as follows.
Given the set $\{X_c^{(t)}\}_c$, we define 
$V_c[\{X_{A|c'}^{(t)}\}_{c'}]$ and $\kappa[\{X_{A|c}^{(t)}\}_c]$
as
\begin{align}
&V_c[\{X_{A|c'}^{(t)}\}_{c'}]
\nonumber \\
:=&\Tr_B \Big(
\log \Big(\sum_{c' \in {\cal C}} X_{A|c'}^{(t)} \otimes \rho_{B|c'}^*\Big)  
\nonumber \\
&\hspace{20ex} -\log \rho_{AB}\Big)
(I_A \otimes \rho_{B|c}^*)
\Big), \\
&\kappa[\{X_{A|c}^{(t)}\}_c]
:=
\sum_c\Tr \exp(X_{A|c}^{(t)}-V_c[\{X_{A|c'}^{(t)}\}_{c'}]).
\end{align}
Then, we choose the updated matrix ${X}_{A|c}^{(t+1)}$
\begin{align} 
{X}_{A|c}^{(t+1)}
:=&
\frac{1}{\kappa[\{X_{A|c}^{(t)}\}_c]}
\exp(X_{A|c}^{(t)}-V_c[\{X_{A|c'}^{(t)}\}_{c'}]).
\end{align}
This algorithm is summarized as Algorithm \ref{AL2}.

\begin{figure}
\begin{algorithm}[H]
\caption{Minimization of $D\Big( \sum_{c}  P(c) \rho_{A|c}\otimes \rho_{B|c}^*
\Big\|\rho_{AB}\Big)$}
\Label{AL2}
\begin{algorithmic}[1]
\State{Choose the initial value $\{X_{A|c}^{(t)}\}_c$;} 
\Repeat
\State{Calculate ${X}_{A|c}^{(t+1)}
:=
\frac{1}{\kappa[\{X_{A|c}^{(t)}\}_c]}
\exp(X_{A|c}^{(t)}-V_c[\{X_{A|c'}^{(t)}\}_{c'}])$;}
\Until{convergence of 
$D\Big( \sum_{c}  X_{A|c}^{(t)}\otimes \rho_{B|c}^*
\Big\|\rho_{AB}\Big)$.} 
\end{algorithmic}
\end{algorithm}
\end{figure}
Then, we have the following theorem.

\begin{theorem}\Label{main}
We choose ${X}_{A|c}^{(1)}$ such that
$\sum_{c \in {\cal C}} |c\rangle \langle c| \otimes {X}_{A|c}^{(1)}$ 
is the completely mixed state $\rho_{AC,\mix}$.
Then, we have
\begin{align}
D\Big( \sum_{c}  X_{A|c}^{(t+1)}\otimes \rho_{B|c}^*
\Big\|\rho_{AB}\Big)
-E_{{\cal D}}^*(\rho_{AB})
\le 
\frac{ \log d_A |{\cal D}|}{t} \Label{XMN3}.
\end{align}
\end{theorem}
The above theorem is shown in Appendix \ref{AP3} by 
using \cite[Theorem 3.3]{RISB}.

\subsection{Algorithm for $E_\alpha^*(\rho_{AB})$}
Next, we consider the algorithm for $E_\alpha^*(\rho_{AB})$.
Similarly, instead of $E_\alpha^*(\rho_{AB})$, we calculate
\begin{align}
E_{\alpha,{\cal D}}^*(\rho_{AB}):=
\min_{ X_{A|c}} 
D_\alpha\Big( \sum_{c}  X_{A|c}\otimes \rho_{B|c}^*
\Big\|\rho_{AB}\Big). \Label{NMT4}
\end{align}

The following algorithm minimizes
$\sgn (\alpha-1)\Tr \Big( \rho_{AB}^{\frac{1-\alpha}{2\alpha}}
\Big(\sum_{c' \in {\cal C}} X_{A|c'} \otimes \rho_{B|c'}^*\Big)  
\rho_{AB}^{\frac{1-\alpha}{2\alpha}}\Big)^{\alpha}$.
For this aim,
we iteratively choose matrices
$X_{A|c}$ as follows.
Given the set $\{X_c^{(t)}\}_c$, we define 
$V_{\alpha,c}[\{X_{A|c'}^{(t)}\}_{c'}]$ and $\kappa_\alpha[\{X_{A|c}^{(t)}\}_c]$
as
\begin{align}
&V_{\alpha,c}[\{X_{A|c'}^{(t)}\}_{c'}]
\nonumber \\
:=&\Tr_B \Big(
\rho_{AB}^{\frac{1-\alpha}{2\alpha}}
\Big( \rho_{AB}^{\frac{1-\alpha}{2\alpha}}
\Big(\sum_{c' \in {\cal C}} X_{A|c'}^{(t)} \otimes \rho_{B|c'}^*\Big)  
\rho_{AB}^{\frac{1-\alpha}{2\alpha}}\Big)^{\alpha-1}
\nonumber \\
&\cdot \rho_{AB}^{\frac{1-\alpha}{2\alpha}}\Big)
(I_A \otimes \rho_{B|c}^*),
 \\
&\kappa_\alpha[\{X_{A|c}^{(t)}\}_c]\nonumber \\
:=&
\sum_c\Tr \exp(X_{A|c}^{(t)}-
\frac{\sgn (\alpha-1)}{\gamma}
V_{\alpha,c}[\{X_{A|c'}^{(t)}\}_{c'}]).
\end{align}
Then, we choose the updated matrix ${X}_{A|c}^{(t+1)}$
\begin{align} 
&{X}_{A|c}^{(t+1)}\nonumber \\
:=&
\frac{1}{\kappa[\{X_{A|c}^{(t)}\}_c]}
\exp(X_{A|c}^{(t)}
-\frac{\sgn (\alpha-1)}{\gamma}
 V_{\alpha,c}[\{X_{A|c'}^{(t)}\}_{c'}]).
\end{align}

\begin{figure}
\begin{algorithm}[H]
\caption{Minimization of 
$\sgn(\alpha-1) e^{(\alpha-1)D_\alpha\Big( \sum_{c}  X_{A|c}\otimes \rho_{B|c}^*
\Big\|\rho_{AB}\Big)}$}
\Label{AL5}
\begin{algorithmic}[1]
\State{Choose the initial value $\{X_{A|c}^{(t)}\}_c$;} 
\Repeat  
\State{Calculate ${X}_{A|c}^{(t+1)}
:=
\frac{1}{\kappa[\{X_{A|c}^{(t)}\}_c]}
\exp(X_{A|c}^{(t)}
-\frac{\sgn (\alpha-1)}{\gamma}
 V_{\alpha,c}[\{X_{A|c'}^{(t)}\}_{c'}])$;}
\Until{convergence of 
$D_\alpha\Big( \sum_{c}  X_{A|c}^{(t)}\otimes \rho_{B|c}^*
\Big\|\rho_{AB}\Big)$.} 
\end{algorithmic}
\end{algorithm}
\end{figure}

Then, we have the following theorem.
\begin{theorem}\Label{mainB}
We choose ${X}_{A|c}^{(1)}$ as 
$\sum_{c \in {\cal C}} |c\rangle \langle c| \otimes {X}_{A|c}^{(1)}$ 
as the completely mixed state $\rho_{AC,\mix}$.
For $\alpha \in (0,1)$, and a sufficiently large real number $\gamma>0$, we have
\begin{align}
&\sgn (\alpha-1)\Tr \Big( \rho_{AB}^{\frac{1-\alpha}{2\alpha}}
\Big(\sum_{c' \in {\cal C}} X_{A|c'}^{(t+1)} \otimes \rho_{B|c'}^*\Big)  
\rho_{AB}^{\frac{1-\alpha}{2\alpha}}\Big)^{\alpha}
\nonumber \\
&-
\sgn(\alpha-1) e^{(\alpha-1)E_{\alpha,{\cal D}}(\rho_{AB})}
\le 
\frac{ \gamma \log d_A |{\cal D}|}{t} \Label{XMN3H}.
\end{align}
\end{theorem}
The above theorem is shown in Appendix \ref{AP3} by 
using \cite[Theorem 3.3]{RISB}.

\section{Application of our algorithm}\Label{S7}
\subsection{Membership problem}\Label{S7-1}
The previous section presents algorithms to calculate 
the proposed measures.
The special case of deciding if the entanglement measure is zero can be used for testing separability in a mathematical sense, which is called 
the membership problem of the set ${\cal S}_{AB}$.
The aim of this section 
is the application of our algorithm to the membership problem of the set ${\cal S}_{AB}$.
That is, we employ the quantity $E^*(\rho_{AB})$ for the membership problem.
This measure will be used for constructing the separation hyperplane.
In the following, we consider how to distinguish the following two cases.
\begin{align}
H_0:& E^*(\rho_{AB})=0 \\
H_1:& E^*(\rho_{AB}) \ge \epsilon_1 .
\end{align}

Due to \eqref{NGR},
when $E_{}^*(\rho_{AB})=0$, we have
\begin{align}
E_{{\cal D}}^*(\rho_{AB})
\le 
\epsilon_2({\cal D})( \| \log \rho_{AB} \| +\log d_A d_B)
+\eta_0 (\epsilon_2({\cal D})).
\end{align}
Also, when $E_{}^*(\rho_{AB})> \epsilon$, we have
\begin{align}
E_{{\cal D}}^*(\rho_{AB}) >\epsilon.
\end{align}
Therefore, if the following two cases are distinguished,
$H_0$ and $H_1$ are distinguished.
\begin{align}
H_0': E_{{\cal D}}^*(\rho_{AB}) \le &
\epsilon_2({\cal D})( \| \log \rho_{AB} \| +\log d_A d_B)
\nonumber \\
&+\eta_0 (\epsilon_2({\cal D})) ,
\\
H_1': E_{{\cal D}}^*(\rho_{AB}) > & \epsilon_1 
\end{align}
under the assumption that 
$\epsilon_2({\cal D})( \| \log \rho_{AB} \| +\log d_A d_B)
+\eta_0 (\epsilon_2({\cal D})) \le \epsilon_1 $.

We choose the initial state according to the assumption of Theorem \ref{main}.
Then, when $H_0'$ holds, i.e., $E_{{\cal D}}^*(\rho_{AB}) \le
\epsilon_2({\cal D})( \| \log \rho_{AB} \| +\log d_A d_B)
+\eta_0 (\epsilon_2({\cal D})) $, we have
\begin{align}
& D\Big( \sum_{c}  X_{A|c}^{(t+1)}\otimes \rho_{B|c}^*
\Big\|\rho_{AB}\Big)
\nonumber \\
\le & \epsilon_2({\cal D})( \| \log \rho_{AB} \| +\log d_A d_B)
+\eta_0 (\epsilon_2({\cal D})) 
+\frac{ \log d_A |{\cal D}|}{t} .
\end{align}
Also, when $E_{{\cal D}}^*(\rho_{AB})> \epsilon$, we have
\begin{align}
D\Big( \sum_{c}  X_{A|c}^{(t+1)}\otimes \rho_{B|c}^*
\Big\|\rho_{AB}\Big)>\epsilon.
\end{align}
Therefore, if the following two cases are distinguished,
$H_0'$ and $H_1'$, i.e. 
$H_0$ and $H_1$ are distinguished.
\begin{align}
H_0'':&
D\Big( \sum_{c}  X_{A|c}^{(t+1)}\otimes \rho_{B|c}^*
\Big\|\rho_{AB}\Big) \nonumber \\
&\le
\epsilon_2({\cal D})( \| \log \rho_{AB} \| +\log d_A d_B)
+\eta_0 (\epsilon_2({\cal D})) \nonumber \\
&\quad +\frac{ \log d_A |{\cal D}|}{t} ,
\\
H_1'':& \Tr \rho_{AC}^{(t+1)}\Omega_F[\rho_{AC}^{(t+1)}]
> \epsilon_1 
\end{align}
under the assumption that 
$\epsilon_2({\cal D})( \| \log \rho_{AB} \| +\log d_A d_B)
+\eta_0 (\epsilon_2({\cal D})) 
+\frac{ \log d_A |{\cal D}|}{t} \le \epsilon_1 $.
That is, the number $t$ of iterations needs to satisfy
$t \ge \log d_A |{\cal D}|/
\Big(\epsilon_1 - 
(\epsilon_2({\cal D})( \| \log \rho_{AB} \| +\log d_A d_B)
+\eta_0 (\epsilon_2({\cal D})) )\Big)$.

Next, we evaluate the computation complexity to distinguish $H_0$ and $H_1$
under the above method
when we choose 
${\cal D}_{n,d_B}$ as ${\cal D}$ in the way as Appendix \ref{AP1}
and we fix $\epsilon_1$ and move $d_A,d_B$ with $d_A \ge d_B$.

In this case, we have
$|{\cal D}_{n,d_B}|\le
 (\frac{2 d_B^{1/2} \pi}{\epsilon_2})^{d_B-1}$ with
$\epsilon_2({\cal D}_{n,d_B})= \epsilon_2 $.
The above condition is rewritten as
\begin{align}
t \ge \frac{\log d_A + (d_B-1)\log \frac{2 d_B^{1/2} \pi}{\epsilon_2}
}
{\epsilon_1 - 
(\epsilon_2( \| \log \rho_{AB} \| +\log d_A d_B)
+\eta_0 (\epsilon_2 ) )}.
\end{align}
Since $\epsilon_1$ is fixed,
we need to choose $\epsilon_2$ to be 
$O(\frac{\epsilon_1}{\log d_Ad_B})$, which realizes
$\epsilon_1 - 
(\epsilon_2( \| \log \rho_{AB} \| +\log d_A d_B)
+\eta_0 (\epsilon_2 ) ) = O(\epsilon_1)$.
The above condition is rewritten as
\begin{align}
t \ge &
\frac
{\log d_A + (d_B-1)\log \frac{2 d_B^{1/2} \pi}{O(\frac{\epsilon_1}{\log d_Ad_B})}
}{O(\epsilon_1)}\nonumber \\
=&
O\Big(\frac
{\log d_A
+(d_B-1)\log \frac{d_B^{1/2} \log d_A d_B}{\epsilon_1}
}{\epsilon_1}\Big)\nonumber \\
=&
O\Big(\frac
{\log d_A
+d_B \log d_B + d_B (\log \log d_A -\log \epsilon_1)
}{\epsilon_1}\Big).
\end{align}

However, each step has various calculations. 
In each step, we need to calculate ${\cal F}[\rho_{AC}]$
from $\rho_{AC}$.
For this calculation,
we need to calculate 
$ \log \rho_{AC}$ and $\Omega_{R,D}[\rho_{AC}]$.
$\rho_{AC}$ is a block diagonal matrix.
That is, it is sufficient to calculate every component.
Every component has 
a matrix function on a $d_Ad_B$-dimensional space,
and a factor whose complexity is linear in 
$|{\cal D}_{n,d_B}| $.
But, the latter part is common for all $c$-components.
The former can be done by diagonalization.
The complexity of the diagonalization on $d$-dimensional matrix 
is $O(d^3)$ with error $\epsilon$ \cite{diag,diag2,diag3}.
Therefore, the calculation complexity of 
${\cal F}[\rho_{AC}]$ from $\rho_{AC}$ is $O(|{\cal D}_{n,d_B}| (d_A^3 d_B^3))
\le O((\frac{2 d_B^{1/2} \pi}{\epsilon_2})^{d_B-1} (d_A^3 d_B^3))
= O((\frac{d_B^{1/2} \log d_A d_B}{\epsilon_1})^{d_B-1} (d_A^3 d_B^3))
$.
The total complexity is 
\begin{align}
O\Big(& \frac
{\log d_A
+d_B \log d_B + d_B (\log \log d_A -\log \epsilon_1)
}{\epsilon_1}\nonumber \\
&\cdot
(\frac{d_B^{1/2} \log d_A d_B}{\epsilon_1})^{d_B-1} (d_A^3 d_B^3)\Big).
\end{align}

On the other hand, the method by \cite[Eq. (38)]{MOP} has the following calculation complexity
\begin{align}
O(d_A^6 \epsilon^{-4d_B}).
\end{align}
In the above evaluation, $\epsilon$ takes the same role as $\epsilon_1$.
For $d_B$, our bound has the factor $d_B^{(d_B-1)/2}$, which is larger than 
$\frac{1}{\epsilon^{4d_B}}$.
However, when $d_B$ is also fixed and only $d_A$ increases,
our bound is simplified to 
$ O(   \max(\log d_A ,-\log\epsilon_1) (\log d_A)^{d_B-1} d_A^3 \epsilon_1^{-d_B})$.
When $\epsilon_1$ and $\epsilon$ are fixed,
our bound is simplified to $O((\log d_A)^{d_B} d_A^3)$,
 and their bound is simplified to
$O(d_A^6)$. 
In this case, our bound is smaller than their bound.

\subsection{Separation hyperplane, i.e.,
entanglement witness}\Label{S7-2}
A linear function to determine a separation hyperplane is called an entanglement witness
in the context of entanglement theory
when the cone is the separable cone. 
Next, we discuss how to obtain a separation hyperplane
by using the above approximation and 
the geometrical characterization by Theorem \ref{TH33}.
This is because 
it is a key step for conic linear programming with a general cone
to find a separation hyperplane for an element that does not belong to the cone \cite{GLS,Goberna,Kalantari,Nueda}.
For this aim, we employ 
the minimizer of $E_{{\cal D}}^*(\rho_{AB})$, i.e., the state
\begin{align}
\sigma_{AB}^*({\cal D}):=
\argmin_{ \sigma_{AB} \in {\cal S}({\cal D})} 
D(\sigma_{AB}\|\rho_{AB}). \Label{NMP9}
\end{align}
Then, the following theorem means that 
the operator $(\log \sigma_{AB}^*({\cal D})- \log \rho_{AB})$ gives
a separation hyperplane
to distinguish $\rho_{AB} $ and $ {\cal S}$.
\begin{theorem}\Label{THBG}
(i) For any separable state $\sigma_{AB}\in {\cal S}$,
we have  
\begin{align}
&\Tr \sigma_{AB} (\log \sigma_{AB}^*({\cal D})- \log \rho_{AB})
\nonumber \\
\ge & D(\sigma_{AB}^*({\cal D}) \|\rho_{AB})
- \|\log \sigma_{AB}^*({\cal D})- \log \rho_{AB}\|
\epsilon_2 ({\cal D}).\Label{NM9}
\end{align}
(ii) Also, we have
\begin{align}
\Tr \rho_{AB} (\log \sigma_{AB}^*({\cal D})- \log \rho_{AB})
=- D(\rho_{AB}\| \sigma_{AB}^*({\cal D}) )<0.\Label{NM6H}
\end{align}
\end{theorem}

\begin{proof}
Since (ii) is trivial, we show only (i) by contradiction.
Since the set ${\cal S}({\cal D}) $ is the convex set,
we have the following in the same way as Theorem \ref{TH33}.
(i) For any separable state $\sigma_{AB}\in {\cal S}({\cal D})$,
we have  
\begin{align}
\Tr \sigma_{AB} (\log \sigma_{AB}^*({\cal D})- \log \rho_{AB})
\ge D(\sigma_{AB}^*({\cal D}) \|\rho_{AB}).\Label{NM92}
\end{align}
Due to the combination of \eqref{NM92} and \eqref{NNU},
any separable state $\sigma_{AB}\in {\cal S}$ satisfies 
\eqref{NM9}.
\end{proof}

\section{Numerical analysis}\Label{S8}
Next, we numerically demonstrate our algorithm
when 
$\cH_A$ is two qubits and $\cH_B$ is one qubit.
We parametrize a qubit pure state as
$|\phi_\theta\rangle:=
\frac{1}{\sqrt{2}}(|0\rangle+e^{\sqrt{-1}\theta }|1\rangle)$.
Consider the entangled state 
$|\Psi_p\rangle:=
\sqrt{p}|\phi_0\rangle|\phi_0\rangle|0\rangle
+\sqrt{1-p}|\phi_{\pi/2}\rangle|\phi_{\pi/2}\rangle
|1\rangle)$.
Given a real number $\lambda$ with $0 <\lambda < 1$,
we focus on the phase-damping channel
${\cal E}_\lambda$ defined as
\begin{align}
{\cal E}_\lambda
\left(\left(
\begin{array}{cc}
a_{0,0} & a_{0,1} \\
a_{1,0} & a_{1,1} 
\end{array}
\right)\right):=
\left(
\begin{array}{cc}
a_{0,0} & \lambda a_{0,1} \\
{\lambda} a_{1,0} & a_{1,1} 
\end{array}
\right).
\end{align}
This channel acts on the state $|\phi_\theta\rangle$ as 
\begin{align}
{\cal E}_\lambda(\rho)=
\frac{1+\lambda}{2}\rho
+
\frac{1-\lambda}{2}
\left(
\begin{array}{cc}
1 & 0 \\
0 & -1
\end{array}
\right)
\rho
\left(
\begin{array}{cc}
1 & 0 \\
0 & -1
\end{array}
\right).
\end{align}
That is, the phase-damping channel is replaced by the above 
over the state $|\Psi_p\rangle$.
Then, we define the state with two parameters 
$\lambda$ and $\delta$ as
\begin{align}
\rho_{p,\lambda,\delta}:=
(1-\delta)
({\cal E}_{\lambda} \otimes {\cal E}_{\lambda} \otimes Id)
(|\Psi_p\rangle
\langle \Psi_p|)
+\delta\rho_{mix},
\end{align}
where $\delta\rho_{mix}$ is the completely mixed state on 
$\cH_A\otimes \cH_B$.

Applying the unitary matrix
$V=\frac{1}{2}\left(
\begin{array}{cc}
1+\sqrt{-1} & 1-\sqrt{-1} \\
1-\sqrt{-1} & 1+\sqrt{-1}
\end{array}
\right)$
on both qubits in $\cH_A$, the above model is converted 
to the following model.
We define $|\tilde{\Psi}_p\rangle
:=\sqrt{p}|\phi_{\pi}\rangle|\phi_{\pi}\rangle|0\rangle
+\sqrt{1-p}|0\rangle|0\rangle|1\rangle)$.
We define 
the channel $\tilde{\cal E}_\lambda$ as
\begin{align}
\tilde{\cal E}_\lambda(\rho):=
\frac{1+\lambda}{2}\rho
+
\frac{1-\lambda}{2}
\left(
\begin{array}{cc}
0 & -1 \\
1 & 0
\end{array}
\right)
\rho
\left(
\begin{array}{cc}
0 & -1 \\
1 & 0
\end{array}
\right).
\end{align}
Then, we define the state with two parameters 
$\lambda$ and $\delta$ as
\begin{align}
\tilde{\rho}_{p,\lambda,\delta}:=
(1-\delta)
\big(\tilde{\cal E}_{\lambda} \otimes \tilde{\cal E}_{\lambda} \otimes Id\big)
(|\tilde{\Psi}_p\rangle\langle \tilde{\Psi}_p|)
+\delta\rho_{mix}.
\end{align}
That is, we have
$(V^{\otimes 2}\otimes I){\rho}_{p,\lambda,\delta}
(V^{\otimes 2}\otimes I)^\dagger=
\tilde{\rho}_{p,\lambda,\delta}$.
Hence, this model is essentially written with real matrix elements.

We make numerical experiments when $n$ is set to be $16$.
That is, $|{\cal D}_{n,d_B}|$ is $16^2$.
By setting $p=1/2$,
Fig. \ref{fig:2} plots the minimum relative entropy
$E_{{\cal D}}^*({\rho}_{1/2,\lambda,\delta})$ 
with $\delta=10^{-1},10^{-2},10^{-4}$.
When $\lambda$ is close to zero,
the state ${\rho}_{p,\lambda,\delta}$ approaches to
the set ${\cal S}_{AB}$ of separable states.

\begin{figure}[ht]
    \centering
    \includegraphics[scale=0.59]{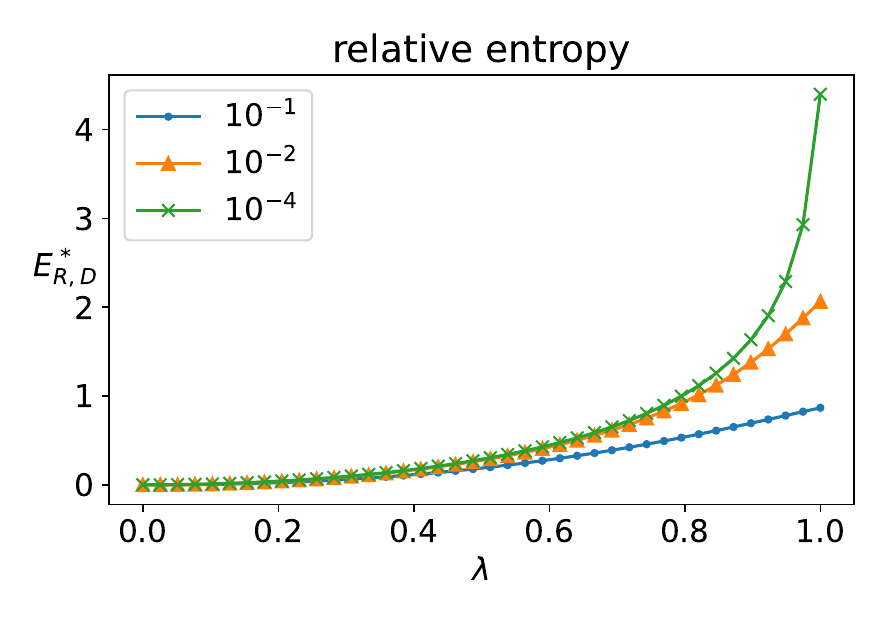}
    \caption{Minimum relative entropy
$E_{{\cal D}}^*({\rho}_{1/2,\lambda,\delta})$:
The vertical axis shows the minimum relative entropy.
The horizontal axis shows the parameter $\lambda$.
Blue (orange, green) curve presents the case with 
$\delta=10^{-1}$, ($10^{-2},10^{-4}$).}
    \label{fig:2}
\end{figure}

Next, we apply our algorithm to the membership problem
of separability by choosing $p=0.1$.
For this aim, ideally, it is needed to plot 
the boundary to satisfy the condition 
$E_{{\cal D}}^*({\rho}_{0.1,\lambda,\delta})=0$.
However, due to the limitation of the performance of our computer, we cannot handle a too small number.
Hence, instead of the ideal value $0$, we set the same small number $10^{-8}$.
That is, Figure \ref{fig:threshold} plots the boundary
to satisfy the condition 
$E_{{\cal D}}^*({\rho}_{0.1,\lambda,\delta})=10^{-8}$.
To compare our method with 
the partial transpose negative test.
This test also has the same problem to handle 
the case with zero eigenvalue.
Instead of the ideal value $0$, we set the same small number $-10^{-8}$ in the same way.
That is, this figure plots the boundary to satisfy the condition that
the minimum eigenvalue of the partial transposed state
is $-10^{-8}$.

These comparisons show that our method can detect entangled states that cannot be detected by 
the partial transpose negative test
even when all the matrix components of the density matrix 
of a $4 \times 2 $ system are written in real numbers. 

\begin{figure}[ht]
    \centering
    \includegraphics[scale=0.54]{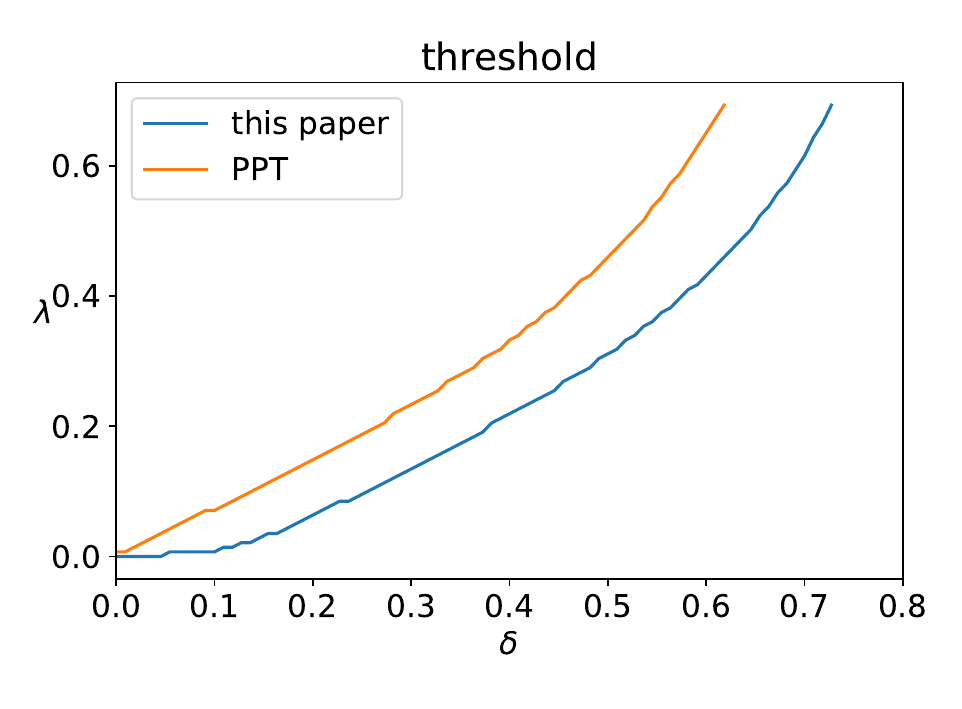}
    \caption{Comparison between our method and the partial transpose negative test: 
    The vertical axis shows the parameter $\lambda$.
    The horizontal axis shows the parameter $\delta$.
    Blue (Orange) curve shows our method (the partial transpose negative test). 
    The upper regions of both curves are composed of states that are guaranteed to be entangled by the respective methods.  
    }
    \label{fig:threshold}
\end{figure}

\section{Exponential quantum Sanov theorem}\Label{ESanov}
Finally, to clarify our contribution on 
quantum Sanov theorem,
we state our obtained result for quantum Sanov theorem in this section.
For this aim,
we consider a $d$-dimensional Hilbert space ${\cal H}$ and a state $\rho$ on $\cH$.
We assume the $n$-fold independently and identical distributed (iid) condition
and consider the following two hypotheses.
\begin{align}
H_0:& \hbox{The state is }\rho^{\otimes n},\\
H_1:& \hbox{The state is }\sigma^{\otimes n} \hbox{ with }\sigma \in S.
\end{align}
To address this problem, we focus on the following quantity.
\begin{align}
\beta_{\epsilon,n}(\rho):=
\min_{0\le T\le I}
\Big\{\Tr T \rho^{\otimes n} \Big|
\max_{\sigma\in S} \Tr (I-T)\sigma^{\otimes n}\le \epsilon
\Big\}.
\end{align}
The relation 
\begin{align}
\lim_{n \to \infty}-\frac{1}{n} \log \beta_{\epsilon,n}(\rho)
=E^*(\rho):=\min_{\sigma \in S}D(\sigma\|\rho)\Label{BVDY}
\end{align}
is known as quantum Sanov theorem \cite{Bjelakovic,Notzel}.
To analyze the case when the value $\epsilon$ goes to zero 
exponentially,
as the extension of $E^*(\rho)$,
we define the minimum sandwiched relative entropy as
\begin{align}
E_{\alpha}^*(\rho)
:=\min_{\sigma \in S}D_{\alpha}(\sigma\|\rho).
\end{align}

As an exponential version of quantum Sanov theorem, the following theorem holds.
\begin{theorem}\Label{thm6}
We have
\begin{align}
\lim_{n \to \infty}
-\frac{1}{n}\log \beta_{e^{-nr},n}(\rho)
&\ge 
\max_{0\le \alpha \le 1}\frac{
(1-\alpha) E_{\alpha}^*(\rho) - \alpha r}{1-\alpha}\Label{NBTY} \\
-\frac{1}{n} \log \beta_{1-e^{-nr},n}(\rho)
&\le 
\min_{\alpha \ge 1}\frac{
(1-\alpha) E_{\alpha}^*(\rho) - \alpha r}{1-\alpha}\Label{NBT2Y}.
\end{align}
More precisely, we have
\begin{align}
&-\frac{1}{n}\log \beta_{e^{-nr},n}(\rho)
\nonumber \\
\ge & 
-\frac{(d+2)(d-1)}{n}\log (n+1)
\nonumber \\
&+\max_{0 \le \alpha \le 1}\frac{(1-\alpha) E_{\alpha}^* (\rho) - \alpha r}{1-\alpha}\Label{NBTX}.
\end{align}
\end{theorem}

Our bound is the optimal 
when the number of elements of $S$ is one and 
the unique element of $S$ is commutative with $\rho$.
Although the paper \cite{Notzel} derived an exponential 
evaluation for both error probabilities,
its exponent does not equal to the optimal bound 
even when the number of elements of $S$ is one and 
the unique element of $S$ is commutative with $\rho$.
In this sense, our evaluation is better than the upper bound by 
\cite{Notzel}.

The proof is given in Appendix \ref{AP4}.
Theorem \ref{thm6} proves Theorem \ref{TH1A}.
Under the limit $r\to 0$,
the RHS of \eqref{NBT} goes to 
$E_{1}^*(\rho)=E^*(\rho)$.
Hence, the relation \eqref{NBT} recovers the inequality $\ge$ in \eqref{BVD}
while the opposite inequality in \eqref{BVD} follows from the converse part of quantum Stein's lemma.

\section{Discussion}\Label{S11}
We have proposed new entanglement measures based on 
the distinguishability of the given entangled state from  
the set of separable states with two kinds of formulations
by using 
quantum relative entropy and sandwiched R\'{e}nyi relative entropy.
That is, we have proposed two kinds of measures based on 
quantum relative entropy.
This measure is different from the relative entropy of entanglement
\cite{Vedral1,Vedral2} because the choice of two input states are opposite to each other.

In the first formulation, the possible separable states are limited to tensor product states in this discrimination problem.
This case can be considered as a special case of quantum Sanov theorem \cite{Bjelakovic,Notzel}. 
Our analysis realizes an exponential evaluation for both kinds of error probabilities, which is better than the existing results \cite{Notzel}.
In the second formulation, 
the possible separable states are not limited to tensor product states.
This problem setting cannot be considered as a special case of
quantum Sanov theorem. 
Even in this setting, we have derived 
an exponential evaluation for both kinds of error probabilities.
Based on these two formulations, we have proposed two kinds of entanglement measures.
The former measure is not smaller than the latter measure.

Also, the proposed quantity is useful for finding entanglement witness
based on its information-geometrical structure.
In the next step, we have derived several useful calculation formulas for these measures.
In particular, we have shown that 
these two kinds of measures coincide with each other
when the entangled state is maximally correlated.
In particular, 
the entanglement measure based on the quantum relative entropy has been concretely derived. 
However, it is not easy to calculate these measures for general entangled states.
To resolve this problem, we have proposed algorithms to calculate 
these measures.
In the first step of this calculation, we have approximated 
these values by introducing the $\delta$-net on the second quantum system, which was introduced to derive an entanglement witness \cite[Section 3]{Ioannou}.
Then, we apply the quantum Arimoto-Blahut algorithm \cite{RISB}
to these approximating values.
Then, we have obtained an efficient algorithm 
for these approximating values.
These algorithms are shown to converge to the global minimum,
and their errors are evaluated in a simple form.  
 
Finally, we have applied our algorithm to the membership problem 
of separability.
Although it is known to be NP-hard 
when the dimensions of both Hilbert spaces increase \cite{Gurvits1,Gurvits2,Gharibian,Ioannou},
the calculation complexity is not so large when 
the dimension of the second Hilbert space is fixed \cite{MOP}.
In this case, our algorithm has a smaller calculation complexity
than the algorithm by \cite{MOP}.
Also, we have applied our algorithm to find an entanglement witness.
 
As a future problem, we can consider the following algorithm. 
We choose a $\delta$ net on the second system $\cH_B$, and apply our algorithm.
 Then, by fixing the obtained states on the system $\cH_A$,
our algorithm exchanges the roles of $\cH_A$ and $\cH_B$ so that we derive states on $\cH_B$.
In this way, we iteratively improve the states on $\cH_A$ and $\cH_B$.
We might expect that this algorithm works well even when  
the $\delta$ is not so small.
However, this algorithm solves 
the membership problem 
of separability, which is known to be NP-hard 
when the dimensions of both Hilbert spaces increase.
Hence, it is impossible to choose 
so small $\delta$ that the calculation time complexity is 
polynomial for the local dimension. 
It is an interesting problem whether such an iterative algorithm can improve
our proposed algorithm. 

\section*{Acknowledgement}
MH was supported in part by the National Natural Science Foundation of China under Grant 62171212.
MH is thankful to Professor Hiroshi Hirai for helpful discussions on the membership problem for separability and
the importance of separation hyperplane,
and informing the references \cite{Gurvits2,DPS2,GLS,MOP,HNW,Fawzi}. 
MH is thankful to Professor Marco Tomamichel and Mr. Roberto Rubboli for 
informing the references \cite{BCY}, \cite{EAP}, \cite{RT24}.
MH is grateful to Dr. Baichu Yu for 
informing the references \cite{TWKK,SCC}.
MH and YI are grateful for Professor Masaki Owari 
to a helpful discussion on the membership problem of separability
and informing the reference \cite{Ioannou}.

\appendices
\section{Construction of discrete subset}\Label{AP1}
\subsection{Real discrete subset}

First, similar to \cite[Section 5.4]{H-O},
we construct the real discrete subset ${\cal D}_{n,d,R}
\subset {\cal S}_d$, where ${\cal S}_d$ is the $d$-dimensional sphere in $\mathbb{R}^{d+1}$.
We define ${\cal D}_{n,1,R}$ as
\begin{align}
{\cal D}_{n,1,R}:=
\Big\{
\Big(\cos \frac{2\pi j}{n},\sin \frac{2\pi j}{n}\Big)\Big\}_{j=0}^{n-1}.
\end{align}
Then, we inductively define ${\cal D}_{n,d,R}$ as
\begin{align}
&{\cal D}_{n,d,R}\notag\\
:=&
\Big\{
\Big( \cos \frac{2\pi j}{n} v ,\sin \frac{2\pi j}{n}\Big)\Big|
v \in {\cal D}_{n,d-1,R},
j=0, \ldots, n-1
\Big\}.
\end{align}

We have
\begin{align}
F({\cal D}_{n,1,R})
:=& \min_{\psi_1\in {\cal S}_2 } \max_{\psi_2\in {\cal D}_{n,2,R}}
|\langle \psi_1|\psi_2\rangle |
= \cos \frac{\pi }{n}.
\end{align}
In general, we have
\begin{align}
F({\cal D}_{n,d,R})
=&\min_{\psi_1\in {\cal S}_d } \max_{\psi_2\in {\cal D}_{n,d,R}}
|\langle \psi_1|\psi_2\rangle|=
\cos^d \frac{\pi }{n}.\Label{ZBX}
\end{align}
The relation \eqref{ZBX} can be inductively shown as \eqref{E105} in the next page.
\begin{figure*}[!t]
\begin{align}
&F({\cal D}_{n,d}) 
= \min_{\psi_1\in {\cal S}_d } \max_{\psi_2\in {\cal D}_{n,d}}
|\langle \psi_1|\psi_2\rangle |\nonumber \\
=& \min_{\theta}\min_{\psi_1\in {\cal S}_{d-1} } 
\max_{j=0,\ldots, n-1}
\max_{\psi_2\in {\cal D}_{n,d-1}}
\Big|
\langle \psi_1|\psi_2\rangle \cos \theta \cos \frac{2\pi j}{n}
+
\sin \theta \sin \frac{2\pi j}{n}\Big|
\nonumber \\
=& \min_{\theta}\min_{\psi_1\in {\cal S}_{d-1} } 
\max_{j=0,\ldots, n-1}
\max_{\psi_2\in {\cal D}_{n,d-1}}
\Big|
\langle \psi_1|\psi_2\rangle 
\Big(\cos \theta \cos \frac{2\pi j}{n}
+
\frac{1}{\langle \psi_1|\psi_2\rangle }
\sin \theta \sin \frac{2\pi j}{n}
\Big)\Big|
\nonumber \\
=& \min_{\theta}\min_{\psi_1\in {\cal S}_{d-1} } 
\max_{\psi_2\in {\cal D}_{n,d-1}}
|\langle \psi_1|\psi_2\rangle |
\Big(\max_{j=0,\ldots, n-1}
\Big|
\cos \theta \cos \frac{2\pi j}{n}
+
\frac{1}{\langle \psi_1|\psi_2\rangle }
\sin \theta \sin \frac{2\pi j}{n}
\Big|\Big)
\nonumber \\
\ge & \min_{\theta}\min_{\psi_1\in {\cal S}_{d-1} } 
\max_{\psi_2\in {\cal D}_{n,d-1}}
|\langle \psi_1|\psi_2\rangle |
\nonumber \\
&
\cdot \Big(\max_{j=0,\ldots, n-1}
\max \Big(
\Big|
\cos \theta \cos \frac{2\pi j}{n}
+
\sin \theta \sin \frac{2\pi j}{n}
\Big|, 
\Big|
\cos \theta \cos \frac{2\pi j}{n}
-
\sin \theta \sin \frac{2\pi j}{n}
\Big|
\Big)
\Big)
\nonumber \\
= & \min_{\theta}\min_{\psi_1\in {\cal S}_{d-1} } 
\max_{\psi_2\in {\cal D}_{n,d-1}}
|\langle \psi_1|\psi_2\rangle |
\Big(\max_{j=0,\ldots, n-1}
\max \Big(
\Big|
\sin (\frac{2\pi j}{n}+\theta )
\Big|,
\Big|
\sin (\frac{2\pi j}{n}-\theta)
\Big|
\Big)
\Big)
\nonumber \\
= & 
\Big(\min_{\theta}
\max_{j=0,\ldots, n-1}
\max \Big(
\Big|
\sin (\frac{2\pi j}{n}+\theta )
\Big|,
\Big|
\sin (\frac{2\pi j}{n}-\theta)
\Big|
\Big)
\Big(\min_{\psi_1\in {\cal S}_{d-1} } 
\max_{\psi_2\in {\cal D}_{n,d-1}}
|\langle \psi_1|\psi_2\rangle|\Big)
\nonumber \\
\ge & \cos \frac{\pi}{n} F({\cal D}_{n,d-1})
= \cos^d \frac{\pi}{n}.\Label{E105}
\end{align}
\end{figure*}

\subsection{Complex discrete subset}
Next, we construct the complex discrete subset ${\cal D}_{n,d}
\subset \mathbb{C}^{d}$, where
all elements of ${\cal D}_{n,d}$ are normalized vectors.
Using ${\cal D}_{n,2d-2,R}$,
we define ${\cal D}_{n,d}$ as
\begin{align}
{\cal D}_{n,d}:=
\{&(x_1, x_2+x_3i,x_4+x_5i,\ldots, x_{2d-2}+x_{2d-1}i)|
\nonumber \\
&x=(x_1,\ldots, x_{2d-1}) \in {\cal D}_{n,2d-2,R}
\}.
\end{align}
The number of elements $|{\cal D}_{n,d}|$
is $ n^{2d-2}$.

Then, we have
\begin{align}
&\frac{1}{4}\epsilon_2({\cal D}_{n,d})^2
=1- \min_{\psi_1\in \mathbb{C}_n^{d}} \max_{\psi_2\in {\cal D}_{n,d}}
|\langle \psi_1|\psi_2\rangle|^2
\nonumber \\
\le &
1- 
(\min_{\psi_1\in \mathbb{C}_n^{d}} \max_{\psi_2\in {\cal D}_{n,d}}
\Re | \langle \psi_1|\psi_2\rangle| )^2
\nonumber \\
=& 1- 
(\min_{\psi_1\in S_{2d-2}} \max_{\psi_2\in {\cal D}_{n,2d-2,R}}
 \langle \psi_1|\psi_2\rangle)^2
\nonumber \\
=& 1- 
(\min_{\psi_1\in S_{2d-2}} \max_{\psi_2\in {\cal D}_{n,2d-2,R}}
 \langle \psi_1|\psi_2\rangle)^2 \nonumber \\
\le & 1-\cos^{2d} \frac{\pi}{n}
= 1-(1-\sin^{2} \frac{\pi}{n})^d
\le d \sin^{2} \frac{\pi}{n}
\le d \frac{\pi^2}{n^2}.
\end{align}

When we choose $n$ as
$ 4d \frac{\pi^2}{n^2}= \epsilon_2^2$, i.e., 
$n=\sqrt{\frac{4 d \pi^2}{\epsilon_2^2}}$,
the number of elements $|{\cal D}_{n,d}|$
is $ (\frac{2 d^{1/2} \pi}{\epsilon_2})^{d-1}$.

\section{Proof of Lemma \ref{L1}}\Label{AP2}
Since the relation \eqref{L1A} follows from 
\eqref{L1B} with taking the limit $\alpha\to 1$, we show only \eqref{L1B}.
Define the projection $P:=
\sum_j |u_{j,A}\rangle|u_{j,B}\rangle \langle u_{j,A}|\langle u_{j,B}|
$
and the TP-CP map
\begin{align}
\Lambda^*(\rho'):=
\sum_{j}
P\rho'P
+(I-P)\rho'(I-P).
\end{align}
We choose a product pure state $|v_A\rangle |v_B\rangle$.
The relation $\|(I-P)|v_A\rangle |v_B\rangle\|=0$
holds if and only if
the relations 
\begin{align}
\begin{split}
&|\{j: |\langle u_{j,A}|v_A\rangle | \neq 0\}
|=1 , \\
&\{j: |\langle u_{j,A}|v_A\rangle | \neq 0\}
= \{j: |\langle u_{j,B}|v_B\rangle | \neq 0\}
\end{split}
\Label{BN5}
\end{align}
hold.
When the above condition holds,
the relation 
$P|v_A\rangle |v_B\rangle
=|u_{j,A}\rangle|u_{j,B}\rangle$
holds for 
$j\in \{j: |\langle u_{j,A}|v_A\rangle | \neq 0\}
\cap \{j: |\langle u_{j,B}|v_B\rangle | \neq 0\}$.

Also, we have
$P|v_A\rangle |u_{j,B}\rangle
=(\langle u_{j,A}|v_A\rangle)
|u_{j,A}\rangle|u_{j,B}\rangle$.
Hence, we have
\begin{align}
&P(\sigma_A \otimes I_B) P=
\sum_{j}
P(\sigma_A \otimes |u_{j,B}\rangle \langle u_{j,B}|) P
\nonumber \\
=&
\sum_{j}
\langle u_{j,A}|\sigma_A  |u_{j,A}\rangle) |u_{j,A}\rangle|u_{j,B}\rangle \langle u_{j,A}|\langle u_{j,B}|.\Label{BM7}
\end{align}
Therefore,
\begin{align}
& C_\alpha^* \Big(\sum_{i,j}\theta_{i,j}|j\rangle \langle i|\Big)
\nonumber \\
=& \min_{p}D_\alpha\Big(\sum_{j}p_{j}|j\rangle \langle j|
\Big\| \sum_{i,j}\theta_{i,j}
|u_{j,A}\rangle|u_{j,B}\rangle \langle u_{i,A}|\langle u_{j,B}|\Big)\nonumber \\
\ge &
E_\alpha^*\Big(
\sum_{i,j}\theta_{i,j}
|u_{j,A}\rangle|u_{j,B}\rangle \langle u_{i,A}|\langle u_{j,B}|\Big).
\Label{BGR}
\end{align}

Now, we fix $\alpha \in (0,1]$.
For any separable state $\sigma_{AB} \in {\cal S}_{AB}$, we have
\begin{align}
\sigma_{AB}
=\sum_{j}p_j \rho_{A,j} \otimes \rho_{B,j}
\le \sum_{j}p_j \rho_{A,j} \otimes I_{B}
=\sigma_A \otimes I_B,
\end{align}
where $\sigma_A:=\Tr_B \sigma_{AB}$.
Hence, 
we have
\begin{align}
\rho_{AB}^{\frac{1-\alpha}{2\alpha}}
\sigma_{AB} \rho_{AB}^{\frac{1-\alpha}{2\alpha}}
\le
\rho_{AB}^{\frac{1-\alpha}{2\alpha}}
(\sigma_A \otimes I_B) \rho_{AB}^{\frac{1-\alpha}{2\alpha}}.
\end{align}
Since $x \mapsto x^\alpha$ is matrix monotone,
we have
\begin{align}
\Tr (\rho_{AB}^{\frac{1-\alpha}{2\alpha}}
\sigma_{AB} \rho_{AB}^{\frac{1-\alpha}{2\alpha}})^\alpha
\le
\Tr (\rho_{AB}^{\frac{1-\alpha}{2\alpha}}
(\sigma_A \otimes I_B) \rho_{AB}^{\frac{1-\alpha}{2\alpha}})^\alpha.
\end{align}
Thus, 
\begin{align}
E_\alpha^*(\rho_{AB})
\ge 
\min_{\sigma_{A} }D_{\alpha}(\sigma_{A} \otimes I_B\|\rho_{AB}).
\Label{BGR2}
\end{align}
We have
\begin{align}
&D_{\alpha} \Big(\sigma_{A} \otimes I_B\Big\|
\sum_{i,j}\theta_{i,j}
|u_{j,A}\rangle|u_{j,B}\rangle \langle u_{i,A}|\langle u_{i,B}|\Big)
\nonumber \\
\ge &
D_{\alpha}\Big(
\Lambda^*(\sigma_{A} \otimes I_B)\Big\|
\sum_{i,j}\theta_{i,j}
|u_{j,A}\rangle|u_{j,B}\rangle \langle u_{i,A}|\langle u_{i,B}|\Big) \nonumber \\
=&
D_{\alpha} \Big(
P(\sigma_A \otimes I_B) P
\Big\|
\sum_{i,j}\theta_{i,j}
|u_{j,A}\rangle|u_{j,B}\rangle \langle u_{i,A}|\langle u_{i,B}|\Big) \nonumber \\
=&
D_{\alpha} \Big(
\sum_{j}
\langle u_{j,A}|\sigma_A  |u_{j,A}\rangle) |u_{j,A}\rangle|u_{j,B}\rangle \langle u_{j,A}|\langle u_{j,B}|
\nonumber \\
&\Big\|
\sum_{i,j}\theta_{i,j}
|u_{j,A}\rangle|u_{j,B}\rangle \langle u_{i,A}|\langle u_{i,B}|\Big) \nonumber \\
=&
D_{\alpha} \Big(
\sum_{j}
\langle u_{j,A}|\sigma_A  |u_{j,A}\rangle) |j\rangle \langle j|
\Big\|
\sum_{i,j}\theta_{i,j}
|j\rangle \langle i|\Big).
\end{align}
Taking the minimum for $\sigma_A$, we have
\begin{align}
&\min_{\sigma_{A} }
D_{\alpha} \Big(\sigma_{A} \otimes I_B\Big\|
\sum_{i,j}\theta_{i,j}
|u_{j,A}\rangle|u_{j,B}\rangle \langle u_{i,A}|\langle u_{i,B}|\Big)
\nonumber \\
\ge &
C_\alpha^* \Big(\sum_{i,j}\theta_{i,j}|j\rangle \langle i|\Big).
\Label{BGR3}
\end{align}
The combination of \eqref{BGR}, \eqref{BGR2}, and \eqref{BGR3} yields \eqref{L1B} with $\alpha \in (0,1]$.

Next, we fix $\alpha>1$.
When a separable state $\sigma_{AB} \in {\cal S}_{AB}$
satisfies $ (I-P)\sigma_{AB}(I-P)\neq 0$,
$
D_{\alpha} \Big(\sigma_{AB}\Big\|
\sum_{i,j}\theta_{i,j}
|u_{j,A}\rangle|u_{j,B}\rangle \langle u_{i,A}|\langle u_{i,B}|\Big)
=\infty$.
Hence, we have
\begin{align}
&E_\alpha^*\Big(\sum_{i,j}\theta_{i,j}
|u_{j,A}\rangle|u_{j,B}\rangle \langle u_{i,A}|\langle u_{i,B}|\Big) \nonumber \\
=&
\min
\Big\{D_{\alpha} \Big(\sigma_{AB}\Big\|
\sum_{i,j}\theta_{i,j}
|u_{j,A}\rangle|u_{j,B}\rangle \langle u_{i,A}|\langle u_{i,B}|\Big) 
\nonumber \\
&\hspace{10ex}\Big|
\sigma_{AB} \in {\cal S}_{AB}, (I-P)\sigma_{AB}(I-P)= 0
\Big\}
\nonumber \\
\stackrel{(a)}{=}&
\min_{p}
D_{\alpha} \Big(
\sum_{j}p_j
|u_{j,A}\rangle|u_{j,B}\rangle \langle u_{j,A}|\langle u_{j,B}|
\nonumber \\
&\hspace{10ex}\Big\|
\sum_{i,j}\theta_{i,j}
|u_{j,A}\rangle|u_{j,B}\rangle \langle u_{i,A}|\langle u_{i,B}|\Big) .
\Label{BNY6}
\end{align}
The relation $(a)$ follows from the condition \eqref{BN5}.
The combination of \eqref{BGR} and \eqref{BNY6} yields \eqref{L1B} with $\alpha >1$.

\section{Proof of Theorems \ref{main} and \ref{mainB}}\Label{AP3}
\subsection{Generalized Arimoto-Blahut algorithm}\Label{}
Extending Arimoto-Blahut algorithm \cite{Arimoto,Blahut},
the paper \cite{RISB} proposed a general algorithm.  
Recently, this general algorithm was generalized to another form
\cite{Iterative}. 
We consider a finite-dimensional Hilbert space ${\cal H}$
and focus on the set ${\cal S}({\cal H})$ of density matrices on ${\cal H}$.
In this section, as a preparation, we consider a minimization problem over 
the set ${\cal S}({\cal H})$.
Given a continuous function $\Omega$ from the set ${\cal S}({\cal H})$ to 
the set of Hermitian matrices on ${\cal H}$, $B({\cal H})$, 
we consider the minimization
$\min_{\rho \in {\cal S}({\cal H})} {\cal G}(\rho)$;
\begin{align}
{\cal G}(\rho):= \Tr \rho \Omega[\rho].
\end{align}
In the following, we discuss the calculation of the following two problems;
\begin{align}
\overline{{\cal G}}(a):=\min_{\rho \in {\cal S}({\cal H})} {\cal G}(\rho), \quad
\rho_{*}:=\argmin_{\rho \in {\cal S}({\cal H})} {\cal G}(\rho).
\end{align}
For this aim, with $\gamma >0$,
we define the density matrix ${\cal F}[\sigma] $ as
${\cal F}[\sigma]:= \frac{1}{\kappa[\sigma]}
\exp( \log \sigma - \frac{1}{\gamma}\Omega[\sigma])$,
where $\kappa[\sigma]$ is the normalization factor
$\Tr \exp( \log \sigma -\frac{1}{\gamma}\Omega[\sigma])$.
The paper \cite{RISB} proposed Algorithm \ref{AL1}.

\begin{figure}
\begin{algorithm}[H]
\caption{Minimization of ${\cal G}(\rho)$}
\Label{AL1}
\begin{algorithmic}[1]
\State{Choose the initial value $\rho^{(1)} \in \mathcal{M}$;} 
\Repeat
\State{Calculate $\rho^{(t+1)}:={\cal F}[\rho^{(t)}]
$;}
\Until{convergence of ${\cal G}(\rho^{(t)})$.} 
\end{algorithmic}
\end{algorithm}
\end{figure}

For any two densities $\rho$ and $\sigma$, we define
\begin{align}
D_\Omega(\rho\|\sigma):=
\Tr \rho (\Omega[\rho]- \Omega[\sigma]).
\end{align}

The paper \cite[Theorem 3.3]{RISB} showed
the following theorem, which discusses the convergence to the global minimum 
and the convergence speed.

\begin{theorem}\Label{TH1}
When any two densities $\rho^{(t)}$ and $\rho^{(t+1)}$ in 
satisfy the condition 
\begin{align}
D_\Omega(\rho^{(t+1)}\|\rho^{(t)})
\le \gamma D(\rho^{(t+1)}\|\rho^{(t)})
\Label{BK1+} ,
\end{align}
and
a state $\rho_* $ satisfies 
\begin{align}
0 \le D_\Omega(\rho_*\|\rho)
\Label{BK2+} 
\end{align}
with any state $\rho$,
Algorithm \ref{AL1} satisfies the condition
\begin{align}
{\cal G}(\rho^{(t_0+1)})
-{\cal G}(\rho_{*})
\le 
\frac{\gamma D(\rho_{*}\| \rho^{(1)}) }{t_0} \Label{XME}
\end{align}
with any initial state $\rho^{(1)}$.
\end{theorem}

\subsection{Proof of Theorem \ref{main}}\Label{Ap3B}
We choose a classical-quantum state $\rho_{AC}$ as
\begin{align}
\rho_{AC}:= 
\sum_{ c \in {\cal C}}
|c\rangle \langle c| \otimes X_{A|c}.
\end{align}
We choose the function $\Omega_{R,{\cal D}}$ as
\begin{align}
&\Omega_{{\cal D}}[\rho_{AC}]
\notag\\
:=&
\sum_{ c \in {\cal C}}
 |c\rangle \langle c| \otimes 
 \Big( 
\Tr_B \Big(
\log \Big(\sum_{c' \in {\cal C}}X_{A|c'} \otimes \rho_{B|c'}^*\Big)  
\notag\\
&-
\log \rho_{AB}\Big)
(I_A \otimes \rho_{B|c}^*)
\Big).
\end{align}
Then, we have
\begin{align}
&\Tr \rho_{AC} \Omega_{\cal D}[\rho_{AC}]\notag\\
=&
\Tr_{AC} \rho_{AC} \sum_{ c \in {\cal C}}
 |c\rangle \langle c| \otimes 
 \Big( 
\Tr_B \Big(
\log \Big(\sum_{c' \in {\cal C}} X_{A|c'} \otimes \rho_{B|c'}^*\Big)  \notag\\
&-
\log \rho_{AB}\Big)
(I_A \otimes \rho_{B|c}^*)
\Big)\nonumber \\
=&
\Tr_{AC} 
(\sum_{ c'' \in {\cal C}}
 |c''\rangle \langle c''| \otimes X_{A|c''})
\sum_{ c \in {\cal C}}
 |c\rangle \langle c| \otimes \notag\\
& \Big( 
\Tr_B \Big(
\log \Big(\sum_{c' \in {\cal C}}X_{A|c'} \otimes \rho_{B|c'}^*\Big)  
-
\log \rho_{AB}\Big)
(I_A \otimes \rho_{B|c}^*)
\Big)\nonumber \\
=&
\Tr_{A} 
\sum_{ c \in {\cal C}}
X_{A|c}
 \Big( 
\Tr_B \Big(
\log \Big(\sum_{c' \in {\cal C}} X_{A|c'} \otimes \rho_{B|c'}^*\Big)  
-
\log \rho_{AB}\Big)\notag\\
&(I_A \otimes \rho_{B|c}^*)
\Big)\nonumber \\
=&
\Tr_{AB} 
\sum_{ c \in {\cal C}}
X_{A|c} \otimes \rho_{B|c}^*
\Big(
\log \Big(\sum_{c' \in {\cal C}} X_{A|c'} \otimes \rho_{B|c'}^*\Big)  \notag\\
&-
\log \rho_{AB}\Big)\notag\\
=& D\Big(\sum_{ c \in {\cal C}}
X_{A|c} \otimes \rho_{B|c}^*\Big\|
 \rho_{AB}\Big).
\end{align}
That is, we can apply 
Algorithm \ref{AL1}.
In this case, when we set 
$\sigma_{AC}[\{X_{A|c}\}_c]$ as
\begin{align}
\sigma_{AC}[\{X_{A|c}\}_c]:= 
\sum_{ c \in {\cal C}}
 |c\rangle \langle c| \otimes X_{A|c},\Label{NMI}
\end{align}
${\cal F}[\sigma_{AC}[\{X_{A|c}\}_c]]$ is written as
\begin{align}
&{\cal F}[\sigma_{AC}[\{X_{A|c}\}_c]]\notag \\
=&
\frac{1}{\kappa[\{X_{A|c}\}_c]}
\sum_{ c \in {\cal C}}
 |c\rangle \langle c| \otimes 
\exp(X_{A|c}-V_c[\{X_{A|c'}\}_{c'}]).
\end{align}
Algorithm \ref{AL2} is the same as the special case of Algorithm \ref{AL1}
with the above choice and $\gamma=1$.

Also, we have
\begin{align}
& D_{\Omega_{{\cal D}}}(\rho_{AC}\|\rho_{AC}')\nonumber \\
=&
\Tr_{AC} \rho_{AC} \sum_{ c \in {\cal C}}
 |c\rangle \langle c| \otimes 
 \Big( 
\Tr_B 
\Big(
\log \Big(\sum_{c' \in {\cal C}}X_{A|c'} \otimes \rho_{B|c'}^*\Big)  \notag\\
&-
\Big(
\log \Big(\sum_{c'' \in {\cal C}} X_{A|c''}' \otimes \rho_{B|c''}^*\Big)  
\Big)
(I_A \otimes \rho_{B|c}^*)
\Big)\nonumber \\
=&
\Tr_{AB} 
(\sum_{c \in {\cal C}}P(c) \rho_{A|c} \otimes \rho_{B|c}^*)
\Big(
\log \Big(\sum_{c' \in {\cal C}} X_{A|c'} \otimes \rho_{B|c'}^*\Big)  \notag\\
&-
\Big(
\log \Big(\sum_{c'' \in {\cal C}} X_{A|c''}' \otimes \rho_{B|c''}^*\Big)  
\Big)\nonumber \\
=& D\Big(\sum_{c' \in {\cal C}} X_{A|c'} \otimes \rho_{B|c'}^*
\Big\|
\sum_{c'' \in {\cal C}} X_{A|c''}' \otimes \rho_{B|c''}^*
\Big) \Label{NBI}\\
\le & D\Big(\sum_{c' \in {\cal C}}
|c'\rangle \langle c'| \otimes 
 X_{A|c'} \otimes \rho_{B|c'}^*
\Big\|
\sum_{c'' \in {\cal C}}
|c''\rangle \langle c''| \otimes \notag\\
& X_{A|c''}' \otimes \rho_{B|c''}^*
\Big)\nonumber \\
= & D\Big(\sum_{c' \in {\cal C}}
|c'\rangle \langle c'| \otimes 
 X_{A|c'} 
\Big\|
\sum_{c'' \in {\cal C}}
|c''\rangle \langle c''| \otimes X_{A|c''}' 
 \Big)\notag\\
=&
D (\rho_{AC}\| \rho_{AC}').
\end{align}
Also, \eqref{NBI} guarantees that the relation 
$D_{\Omega_{R,{\cal D}}}(\rho_{AC}\|\rho_{AC}')\ge 0$ holds.
Hence, the condition of Theorem \ref{TH1} holds.
That is, Theorem \ref{main} holds.

\subsection{Proof of Theorem \ref{mainB}}\Label{Ap3C}
We choose the function $\Omega_{\alpha,{\cal D}}$ as
\begin{align}
&\Omega_{\alpha,{\cal D}}[\rho_{AC}]\notag \\
:=&
\sgn (\alpha-1)\sum_{ c \in {\cal C}}
 |c\rangle \langle c| \otimes \notag\\
&\Tr_{B} 
 \Big( 
\rho_{AB}^{\frac{1-\alpha}{2\alpha}}
\Big( \rho_{AB}^{\frac{1-\alpha}{2\alpha}}
\Big(\sum_{c' \in {\cal C}} X_{A|c'} \otimes \rho_{B|c'}^*\Big)  
\rho_{AB}^{\frac{1-\alpha}{2\alpha}}\Big)^{\alpha-1}
\rho_{AB}^{\frac{1-\alpha}{2\alpha}}\Big)\notag\\
&(I_A \otimes \rho_{B|c}^*)
\Big).
\end{align}
Then, we have
\begin{align}
&\Tr \rho_{AC} \Omega_{\alpha,\cal D}[\rho_{AC}]\nonumber \\
=&
\sgn (\alpha-1)\Tr_{AC} \rho_{AC} 
\sum_{ c \in {\cal C}}
 |c\rangle \langle c| \otimes \notag\\
&\Tr_{B}  \Big( 
\rho_{AB}^{\frac{1-\alpha}{2\alpha}}
\Big( \rho_{AB}^{\frac{1-\alpha}{2\alpha}}
\Big(\sum_{c' \in {\cal C}} X_{A|c'} \otimes \rho_{B|c'}^*\Big)  
\rho_{AB}^{\frac{1-\alpha}{2\alpha}}\Big)^{\alpha-1}
\rho_{AB}^{\frac{1-\alpha}{2\alpha}}\Big)\notag\\
&(I_A \otimes \rho_{B|c}^*)
\Big)
\nonumber \\
=&
\sgn (\alpha-1)\Tr_{AC} 
(\sum_{ c'' \in {\cal C}}
 |c''\rangle \langle c''| \otimes X_{A|c''})
\sum_{ c \in {\cal C}}
 |c\rangle \langle c| \nonumber \\
& \otimes
\Tr_{B} 
 \Big( 
\rho_{AB}^{\frac{1-\alpha}{2\alpha}}
\Big( \rho_{AB}^{\frac{1-\alpha}{2\alpha}}
\Big(\sum_{c' \in {\cal C}} X_{A|c'} \otimes \rho_{B|c'}^*\Big)  
\rho_{AB}^{\frac{1-\alpha}{2\alpha}}\Big)^{\alpha-1}
\rho_{AB}^{\frac{1-\alpha}{2\alpha}}\Big)\notag\\
&(I_A \otimes \rho_{B|c}^*)
\Big)
\nonumber \\
=&
\sgn (\alpha-1)\Tr
(\sum_{ c \in {\cal C}}
 X_{A|c}\otimes \rho_{B|c}^*)\notag\\
& \Big( 
\rho_{AB}^{\frac{1-\alpha}{2\alpha}}
\Big( \rho_{AB}^{\frac{1-\alpha}{2\alpha}}
\Big(\sum_{c' \in {\cal C}} X_{A|c'} \otimes \rho_{B|c'}^*\Big)  
\rho_{AB}^{\frac{1-\alpha}{2\alpha}}\Big)^{\alpha-1}
\rho_{AB}^{\frac{1-\alpha}{2\alpha}}\Big)
\Big)
\nonumber \\
=&
\sgn (\alpha-1)\Tr 
\Big( \rho_{AB}^{\frac{1-\alpha}{2\alpha}}
\Big(\sum_{c' \in {\cal C}} X_{A|c'} \otimes \rho_{B|c'}^*\Big)  
\rho_{AB}^{\frac{1-\alpha}{2\alpha}}\Big)^{\alpha}
\Big).
\end{align}
That is, we can apply 
Algorithm \ref{AL1}.
In this case, when we set 
$\sigma_{AC}[\{X_{A|c}\}_c]$ as \eqref{NMI},
${\cal F}[\sigma_{AC}[\{X_{A|c}\}_c]]$ is written as 
$\frac{1}{\kappa[\{X_{A|c}\}_c]}
\exp(X_{A|c}
-\frac{\sgn (\alpha-1)}{\gamma} V_{\alpha,c}[\{X_{A|c'}^{(t)}\}_{c'}])$.
Algorithm \ref{AL5} is the same as the special case of Algorithm \ref{AL1}
with the above choice.

We denote the minimizer by 
$\sigma_{AC}[\{X_{A|c}^*\}_c]$.
For $\alpha \in [0,1]$,
We choose 
$Z_*=\Big( \rho_{AB}^{\frac{1-\alpha}{2\alpha}}
\Big(\sum_{c \in {\cal C}} X_{A|c}^* \otimes \rho_{B|c}^*\Big)  
\rho_{AB}^{\frac{1-\alpha}{2\alpha}}\Big)
$
and $Z=\Big( \rho_{AB}^{\frac{1-\alpha}{2\alpha}}
\Big(\sum_{c' \in {\cal C}} X_{A|c'} \otimes \rho_{B|c'}^*\Big)  
\rho_{AB}^{\frac{1-\alpha}{2\alpha}}\Big)
$. Since $\sigma_{AC}[\{X_{A|c}^*\}_c]$ is the minimizer,
we have
\begin{align}
\Tr Z_*^\alpha \ge \Tr Z^\alpha.
\end{align}
Since 
$Z_*^\alpha /\Tr Z_*^\alpha$ and 
$(Z^\alpha /\Tr Z^\alpha)$ are density matrices, 
we have
\begin{align}
\Tr(Z_*^\alpha /\Tr Z_*^\alpha)^{\frac{1}{\alpha}} 
(Z^\alpha /\Tr Z^\alpha)^{\frac{\alpha-1}{\alpha}}
\ge 1.
\end{align}
Thus, we have
\begin{align}
& D_{\Omega_{\alpha,{\cal D}}}(\sigma_{AC}[\{X_{A|c}^*\}_c]\|
\sigma_{AC}[\{X_{A|c}\}_c])\nonumber \\
=&
-\Tr_{AC} \sigma_{AC}[\{X_{A|c}^*\}_c] \sum_{ c \in {\cal C}}
 |c\rangle \langle c| \otimes \notag\\
& \Big( 
\Tr_B  \Big( 
 \Big( 
\rho_{AB}^{\frac{1-\alpha}{2\alpha}}
\Big( \rho_{AB}^{\frac{1-\alpha}{2\alpha}}
\Big(\sum_{c \in {\cal C}} X_{A|c}^* \otimes \rho_{B|c}^*\Big)  
\rho_{AB}^{\frac{1-\alpha}{2\alpha}}\Big)^{\alpha-1}
\rho_{AB}^{\frac{1-\alpha}{2\alpha}}\Big) \nonumber \\
&-
 \Big( 
\rho_{AB}^{\frac{1-\alpha}{2\alpha}}
\Big( \rho_{AB}^{\frac{1-\alpha}{2\alpha}}
\Big(\sum_{c' \in {\cal C}} X_{A|c'} \otimes \rho_{B|c'}^*\Big)  
\rho_{AB}^{\frac{1-\alpha}{2\alpha}}\Big)^{\alpha-1}
\rho_{AB}^{\frac{1-\alpha}{2\alpha}}\Big)
\Big)\notag\\
&(I_A \otimes \rho_{B|c}^*)
\Big)\nonumber \\
=&
-\Tr
\Big( \rho_{AB}^{\frac{1-\alpha}{2\alpha}}
\Big(\sum_{c \in {\cal C}} X_{A|c}^* \otimes \rho_{B|c}^*\Big)  
\rho_{AB}^{\frac{1-\alpha}{2\alpha}}\Big)\notag\\
&\Bigg(
\Big( \rho_{AB}^{\frac{1-\alpha}{2\alpha}}
\Big(\sum_{c \in {\cal C}} X_{A|c}^* \otimes \rho_{B|c}^*\Big)  
\rho_{AB}^{\frac{1-\alpha}{2\alpha}}\Big)^{\alpha-1}\nonumber \\
&-
\Big( \rho_{AB}^{\frac{1-\alpha}{2\alpha}}
\Big(\sum_{c' \in {\cal C}} X_{A|c'} \otimes \rho_{B|c'}^*\Big)  
\rho_{AB}^{\frac{1-\alpha}{2\alpha}}\Big)^{\alpha-1}
\Bigg)\nonumber \\
=&\Tr Z_* Z^{\alpha-1} -\Tr Z_*^\alpha \nonumber \\
=& 
( \Tr Z_* Z^{\alpha-1})
(\Tr Z_*^\alpha)^{\frac{\alpha-1}{\alpha}}
((\Tr Z_*^\alpha)^{\frac{1-\alpha}{\alpha}} -
(\Tr Z^\alpha)^{\frac{1-\alpha}{\alpha}} )\notag\\
&+
\Tr Z_*^\alpha (
\Tr(Z_*^\alpha /\Tr Z_*^\alpha)^{\frac{1}{\alpha}} 
(Z^\alpha /\Tr Z^\alpha)^{\frac{\alpha-1}{\alpha}}
-1) \nonumber \\
\ge & 0.
\end{align}
Hence, the condition \eqref{BK2+} holds.

We choose 
$Z'=\Big( \rho_{AB}^{\frac{1-\alpha}{2\alpha}}
\Big(\sum_{c \in {\cal C}} X_{A|c}' \otimes \rho_{B|c}^*\Big)  
\rho_{AB}^{\frac{1-\alpha}{2\alpha}}\Big)$.
Then, we have
\begin{align}
& D_{\Omega_{\alpha,{\cal D}}}(\sigma_{AC}[\{X_{A|c}'\}_c]\|
\sigma_{AC}[\{X_{A|c}\}_c])\\
=& 
( \Tr Z' Z^{\alpha-1})
(\Tr {Z'}^\alpha)^{\frac{\alpha-1}{\alpha}}
((\Tr {Z'}^\alpha)^{\frac{1-\alpha}{\alpha}} -
(\Tr Z^\alpha)^{\frac{1-\alpha}{\alpha}} ) \\
&+
\Tr {Z'}^\alpha (
\Tr({Z'}^\alpha /\Tr {Z'}^\alpha)^{\frac{1}{\alpha}} 
(Z^\alpha /\Tr Z^\alpha)^{\frac{\alpha-1}{\alpha}}
-1) .
\end{align}
When $Z' $ is close to $Z$,
$\Tr({Z'}^\alpha /\Tr {Z'}^\alpha)^{\frac{1}{\alpha}} 
(Z^\alpha /\Tr Z^\alpha)^{\frac{\alpha-1}{\alpha}}
-1
=O(\|Z'-Z\|^2)$
and
$(\Tr {Z'}^\alpha)^{\frac{1-\alpha}{\alpha}} -
(\Tr Z^\alpha)^{\frac{1-\alpha}{\alpha}} =O(\|Z'-Z\|)$.
However, when ${Z'}$ and $Z$ are close to $X_*$,
the first order derivative of $\Tr {Z}^\alpha$ is zero.
Hence, 
$(\Tr {Z'}^\alpha)^{\frac{1-\alpha}{\alpha}} -
(\Tr Z^\alpha)^{\frac{1-\alpha}{\alpha}} =O(\|Z'-Z\|^2)$.
Since $D(\sigma_{AC}[\{X_{A|c}'\}_c]\|
\sigma_{AC}[\{X_{A|c}\}_c])=O(\|Z'-Z\|^2)$,
the value
$\frac{D_{\Omega_{\alpha,{\cal D}}}(\sigma_{AC}[\{X_{A|c}'\}_c]\|
\sigma_{AC}[\{X_{A|c}\}_c])}{D(\sigma_{AC}[\{X_{A|c}'\}_c]\|
\sigma_{AC}[\{X_{A|c}\}_c])}$
takes a finite value when $\sigma_{AC}[\{X_{A|c}'\}_c]$ and $\sigma_{AC}[\{X_{A|c}\}_c]$ are close to $\sigma_{AC}[\{X_{A|c}^*\}_c]$.
Since $\{\sigma_{AC}[\{X_{A|c}^{(t)}\}_c]\}$ converges to an element,
$\frac{D_{\Omega_{\alpha,{\cal D}}}(\sigma_{AC}[\{X_{A|c}^{(t+1)}\}_c]\|
\sigma_{AC}[\{X_{A|c}^{(t)}\}_c])}{D(\sigma_{AC}[\{X_{A|c}^{(t+1)}\}_c]\|
\sigma_{AC}[\{X_{A|c}^{(t)}\}_c])}$
can be bounded by a finite value.
That is, when $\gamma>0$ is taken to be a sufficiently large value,
the condition \eqref{BK1+} holds.

\section{Proof of Theorem \ref{thm6}}\Label{AP4}
\subsection{Proof of \eqref{NBTY}}
Similar to the papers \cite{H-01,H-02,Notzel},
we employ Schur duality of $\cH^{\otimes n}$ as
\begin{align}
\cH^{\otimes n}=
\bigoplus_{\lambda \in Y_d^n}
{\cal U}_\lambda \otimes {\cal V}_\lambda,
\end{align}
${\cal U}_\lambda$ expresses the irreducible space of $\SU(d)$
and ${\cal V}_\lambda$ expresses the irreducible space of  
the representation $\pi$ of the permutation group $\frS_n$.
Then, for any state $\rho$, we have
\begin{align}
\rho^{\otimes n}=
\bigoplus_{\lambda \in Y_d^n}
\rho_\lambda \otimes \rho_{\lambda,mix}, \Label{BBR}
\end{align}
where $\rho_{\lambda,mix}$ is the completely mixed state on ${\cal V}_\lambda$.
Since the system ${\cal V}_\lambda$ has no information, it is sufficient to handle the states
$\oplus_{\lambda \in Y_d^n}
\rho_\lambda ,
\oplus_{\lambda \in Y_d^n}
\sigma_\lambda 
$ on $\oplus_{\lambda \in Y_d^n}
{\cal U}_\lambda$.
Here, we define the map $\kappa_n$ as
\begin{align}
\kappa_n(\rho'):= \sum_{\lambda \in Y_d^n} 
\Tr_{{\cal V}_\lambda} 
P({\cal U}_\lambda \otimes {\cal V}_\lambda)  \rho'P({\cal U}_\lambda \otimes {\cal V}_\lambda)
\end{align}
for a general state on $\cH^{\otimes n}$,
where $P({\cal K})$ expresses the projection to the subspace ${\cal K}$.
We choose a basis 
$\{|u_{\lambda,j}\rangle\}_{j=1}^{d_\lambda}$ 
of ${\cal U}_\lambda$
such that $|u_{\lambda,j}\rangle$ is an eigenvector of $\rho_\lambda $,
where $d_\lambda:= \dim {\cal U}_\lambda$.
We define the pinching map $\Gamma_{\rho,n}$ as
\begin{align}
&\Gamma_{\rho,n}(\rho')\nonumber \\
:=&\sum_{\lambda\in Y_d^n, j}
(|u_{\lambda,j}\rangle \langle u_{\lambda,j}| \otimes P( {\cal V}_\lambda))
\rho'
(|u_{\lambda,j}\rangle \langle u_{\lambda,j}| \otimes P( {\cal V}_\lambda))
\end{align}
for a general state $\rho'$ on $\cH^{\otimes n}$.
Due to the above mentioned structure,
states $\Gamma_{\rho,n}(\sigma^{\otimes n})$ and $\rho^{\otimes n}$
are commutative with each other.
We denote the dimension of 
$\oplus_{\lambda \in Y_d^n}
{\cal U}_\lambda$ by $d_n$.
Since $\dim {\cal U}_\lambda \le (n+1)^{d(d-1)/2}$
and $|Y_d^n|\le (n+1)^{d-1} $ \cite[(6.16) and (6.18)]{H-q-text},
we have 
\begin{align}
d_n \le  (n+1)^{(d+2)(d-1)/2}.\Label{BVT}
\end{align}
We have the pinching inequality:
\begin{align}
\rho'\le d_n \Gamma_{\rho,n}(\rho').\Label{NTD}
\end{align}

We choose 
\begin{align}
T_n(\sigma,R):= 
\{  \Gamma_{\rho,n}(\sigma^{\otimes n}) \ge e^{n R} \rho^{\otimes n}
\}.\Label{VCX}
\end{align}
Given $\alpha \in [0,1]$,
we choose $R_{n,\alpha,r}:= 
\frac{-\max_{\sigma \in S}
\log \Tr \Gamma_{\rho,n}(\sigma^{\otimes n})^{\alpha} (\rho^{\otimes n})^{1-\alpha}-n r}{n(1-\alpha)}
$ so that
$-n r= \max_{\sigma \in S}
\log \Tr \Gamma_{\rho,n}(\sigma^{\otimes n})^{\alpha} (\rho^{\otimes n})^{1-\alpha}
+(1-\alpha) n R_{n,\alpha,r}$.
Here, we address several projections that are commutative with each other.
For two commutative projections $P_1$ and  $P_2$, 
we define 
\begin{align}
P_1\cup P_2:= I-(I-P_1)(I-P_2).\Label{ASR}
\end{align}
Since it is natural to define $(I-P_1)\cap (I-P_2):=(I-P_1)(I-P_2)$,
the above definition for $P_1\cup P_2$ is also natural.
When $\{|u_j\rangle\}_j$ is a common eigenvector system of 
$P_1$ and $P_2$, they can be written as
$P_k=\sum_{j \in J_k}|u_j \rangle \langle u_j|$.
Then, we have 
$P_1\cup P_2=\sum_{j \in J_1 \cup J_2}|u_j \rangle \langle u_j|$.
Based in the definition \eqref{ASR},
we define the operator 
$T_{n,\alpha,r}$ as
\begin{align}
T_{n,\alpha,r}:=\bigcup_{\sigma \in S} T_n(\sigma,R_{n,\alpha,r}).
\end{align}

Although $T_n(\sigma,R_{n,\alpha,r})$ is an operator on $\cH^{\otimes n}$,
it can be written by using a projection operator 
$P_{\sigma,n,\lambda}$
on 
${\cal U}_\lambda$
as
\begin{align}
T_n(\sigma,R_{n,\alpha,r})
=\bigoplus_{\lambda \in Y_d^n}P_{\sigma,n,\lambda}\otimes P( {\cal V}_\lambda).
\end{align}
That is, we choose the projection operator 
$P_{\sigma,n,\lambda}$ as the above.
Hence, we have
\begin{align}
T_{n,\alpha,r}=\bigoplus_{\lambda \in Y_d^n}
\Big(\bigcup_{\sigma \in S} P_{\sigma,n,\lambda}\Big)\otimes 
P( {\cal V}_\lambda).
\Label{BSO}
\end{align}
Since 
all projections $P_{\sigma,n,\lambda}$ are diagonalized with 
the basis $\{|u_{\lambda,j}\rangle\}_{j=1}^{d_\lambda}$ 
of ${\cal U}_\lambda$,
it is sufficient to clarify whether 
the range of $P_{\sigma,n,\lambda}$ contains
$|u_{\lambda,j}\rangle$.
If the range of $P_{\sigma,n,\lambda}$ contains $|u_{\lambda,j}\rangle$,
there exists an element $\sigma_{\lambda,j} \in S$ such that 
$P_{\sigma_{\lambda,j},n,\lambda}$ contains $|u_{\lambda,j}\rangle$.
If the range of $P_{\sigma,n,\lambda}$ does not contain $|u_{\lambda,j}\rangle$,
no element $\sigma \in S$ satisfy the condition 
$P_{\sigma,n,\lambda}$ contains $|u_{\lambda,j}\rangle$.
Therefore, we can choose $d_\lambda$ elements $\sigma_{\lambda,1}, \ldots, 
\sigma_{\lambda,d_\lambda}$
in $S$ such that
\begin{align}
\bigcup_{\sigma \in S} P_{\sigma,n,\lambda}
=\bigcup_{j\in [d_\lambda]} P_{\sigma_{\lambda,j},n,\lambda},
\end{align}
where $[d_\lambda]:=\{1, \ldots, d_\lambda\} $.
Here, the elements $\sigma_{\lambda,j}$ are not necessarily distinct, in general.

Therefore, we have
\begin{align}
&T_{n,\alpha,r}
=\bigoplus_{\lambda \in Y_d^n}
\Big(\bigcup_{j\in [d_\lambda]} P_{\sigma_{\lambda,j},n,\lambda}\Big)
\otimes P( {\cal V}_\lambda)
\nonumber \\
=& \bigcup_{\lambda' \in Y_d^n, j\in [d_\lambda]}
\bigoplus_{\lambda \in Y_d^n}
( P_{\sigma_{\lambda',j},n,\lambda})\otimes P( {\cal V}_\lambda)
\nonumber \\
=& \bigcup_{\lambda' \in Y_d^n, j\in [d_\lambda]} T_n(\sigma_{\lambda',j},R_{n,\alpha,r}).\Label{BCK}
\end{align}
We have
\begin{align}
& \Tr \rho^{\otimes n} T_{n,\alpha,r}
\stackrel{(a)}{\le }
\Tr \rho^{\otimes n} 
\sum_{\lambda' \in Y_d^n, j\in [d_\lambda]} T_n(\sigma_{\lambda',j},
R_{n,\alpha,r})\nonumber \\
= & \sum_{\lambda' \in Y_d^n, j\in [d_\lambda]}
\Tr \rho^{\otimes n} 
T_n(\sigma_{\lambda',j},R_{n,\alpha,r}) \nonumber \\
\stackrel{(b)}{\le} &\sum_{\lambda' \in Y_d^n, j\in [d_\lambda]}
\Tr \Gamma_{\rho,n}(\sigma_{\lambda',j}^{\otimes n})^{\alpha} (\rho^{\otimes n})^{1-\alpha}
e^{-n\alpha R_{n,\alpha,r}} 
\nonumber \\
\stackrel{(c)}{\le} & d_{n}
\max_{\sigma \in S}
\Tr \Gamma_{\rho,n}(\sigma^{\otimes n})^{\alpha} (\rho^{\otimes n})^{1-\alpha}
e^{-n \alpha R_{n,\alpha,r}} \nonumber \\
\stackrel{(d)}{=} & d_{n}
\big(\max_{\sigma \in S} \Tr \Gamma_{\rho,n}(\sigma^{\otimes n})^\alpha 
(\rho^{\otimes n})^{1-\alpha} \big)
\nonumber \\
&\cdot 
\big(\max_{\sigma \in S} \Tr \Gamma_{\rho,n}(\sigma^{\otimes n})^\alpha (\rho^{\otimes n})^{1-\alpha}
e^{nr}\big)^{\frac{\alpha}{1-\alpha}}
\nonumber \\
= & d_{n}
(\max_{\sigma \in S} 
\Tr \Gamma_{\rho,n}(\sigma^{\otimes n})^{\alpha} 
(\rho^{\otimes n})^{1-\alpha})^{\frac{1}{1-\alpha}}
e^{n\frac{\alpha r}{1-\alpha} }.\Label{NMR1}
\end{align}
Each step can be shown as follows.
Step $(a)$ follows from \eqref{BCK}.
Step $(b)$ can be shown as follows.
Since the condition $\Gamma_{\rho,n}(\sigma^{\otimes n}) \ge e^{n R} \rho^{\otimes n}$ implies
$\Gamma_{\rho,n}(\sigma^{\otimes n})^{\alpha} \ge e^{\alpha n R} 
(\rho^{\otimes n})^{\alpha}$,
the definition \eqref{VCX} guarantees that
the inequality
\begin{align*}
&T_n(\sigma_{\lambda',j},R_{n,\alpha,r})\rho^{\otimes n}\\
=&T_n(\sigma_{\lambda',j},R_{n,\alpha,r})
(\rho^{\otimes n})^{1-\alpha}
(\rho^{\otimes n})^{\alpha} \\
\le&
\Tr T_n(\sigma_{\lambda',j},R_{n,\alpha,r})
\Gamma_{\rho,n}(\sigma_{\lambda',j}^{\otimes n})^{\alpha} (\rho^{\otimes n})^{1-\alpha}
e^{-n\alpha R_{n,\alpha,r}} \\
=&\Tr 
\Gamma_{\rho,n}(\sigma_{\lambda',j}^{\otimes n})^{\alpha} (\rho^{\otimes n})^{1-\alpha}
e^{-n\alpha R_{n,\alpha,r}} .
\end{align*}
Step $(c)$ follows from the relation $
\sum_{\lambda \in Y_d^n} d_\lambda= \dim 
\oplus_{\lambda \in Y_d^n}{\cal U}_\lambda =d_n$.
Step $(d)$ follows from the definition of $R_{n,\alpha,r}$.

For 
$\sigma \in S$,
we have
\begin{align}
&\Tr \sigma^{\otimes n} (I-T_{n,\alpha,r})
=\Tr \sigma^{\otimes n} \Gamma_{\rho,n}(I-T_{n,\alpha,r})
\nonumber \\
=& \Tr \Gamma_{\rho,n}(\sigma^{\otimes n}) (I-T_{n,\alpha,r})
\nonumber \\
\stackrel{(a)}{\le} & \Tr \Gamma_{\rho,n}(\sigma^{\otimes n}) (I-
T_n(\sigma,R_{n,\alpha,r})) \nonumber \\
\stackrel{(b)}{\le} & \Tr \Gamma_{\rho,n}(\sigma^{\otimes n})^{\alpha} (\rho^{\otimes n})^{1-\alpha}
e^{n(1-\alpha) R_{n,\alpha,r}}
\nonumber \\
\le & \max_{\sigma \in S}
\Tr \Gamma_{\rho,n}(\sigma^{\otimes n})^{\alpha} (\rho^{\otimes n})^{1-\alpha}
e^{n(1-\alpha) R_{n,\alpha,r}}
=e^{-nr}.\Label{NBC1}
\end{align}
Each step can be shown as follows.
Step $(a)$ follows from the relation
$T_n(\sigma,R_{n,\alpha,r}) \le 
T_{n,\alpha,r}$, which is shown by the relation \eqref{BSO}.
Step $(b)$ follows from \eqref{VCX} in the same way as
Step $(b)$ of \eqref{NMR1}.
Thus, we have
\begin{align}
-\frac{1}{n}\log \beta_{e^{-nr},n}(\rho)
\ge -\frac{1}{n}\log  \Tr \rho^{\otimes n} T_{n,\alpha,r}.\Label{MM1}
\end{align}

In fact, applying the matrix monotone function to \eqref{NTD},
we have
\begin{align}
&d_n^{\alpha-1} 
((\rho^{\otimes n})^{\frac{1-\alpha}{2\alpha}}
\Gamma_{\rho,n}(\sigma^{\otimes n})
(\rho^{\otimes n})^{\frac{1-\alpha}{2\alpha}})^{\alpha-1}
\nonumber \\
=&
\Big(d_n
\Gamma_{\rho,n}\Big(
(\rho^{\otimes n})^{\frac{1-\alpha}{2\alpha}}
\sigma^{\otimes n}
(\rho^{\otimes n})^{\frac{1-\alpha}{2\alpha}} \Big)
\Big)^{\alpha-1} \nonumber \\
\le &
\Big(
(\rho^{\otimes n})^{\frac{1-\alpha}{2\alpha}}
\sigma^{\otimes n}
(\rho^{\otimes n})^{\frac{1-\alpha}{2\alpha}} 
\Big)^{\alpha-1}.
\Label{NTD5}
\end{align}
Thus, we have
\begin{align}
&\Tr \Gamma_{\rho,n}(\sigma^{\otimes n})^{\alpha} 
(\rho^{\otimes n})^{1-\alpha}
\nonumber \\
=&
\Tr 
((\rho^{\otimes n})^{\frac{1-\alpha}{2\alpha}}
\Gamma_{\rho,n}(\sigma^{\otimes n})
(\rho^{\otimes n})^{\frac{1-\alpha}{2\alpha}})^\alpha \nonumber \\
=&
\Tr 
\Gamma_{\rho,n}(
(\rho^{\otimes n})^{\frac{1-\alpha}{2\alpha}}
\sigma^{\otimes n}
(\rho^{\otimes n})^{\frac{1-\alpha}{2\alpha}}))^\alpha \nonumber \\
=&
\Tr
\big(
\Gamma_{\rho,n}(
(\rho^{\otimes n})^{\frac{1-\alpha}{2\alpha}}
\sigma^{\otimes n}
(\rho^{\otimes n})^{\frac{1-\alpha}{2\alpha}}) \big)
\nonumber \\
&\cdot
 \Big(
\Gamma_{\rho,n}(
(\rho^{\otimes n})^{\frac{1-\alpha}{2\alpha}}
\sigma^{\otimes n}
(\rho^{\otimes n})^{\frac{1-\alpha}{2\alpha}})\Big)^{\alpha-1} \nonumber \\
=&
\Tr 
((\rho^{\otimes n})^{\frac{1-\alpha}{2\alpha}}
\sigma^{\otimes n}
(\rho^{\otimes n})^{\frac{1-\alpha}{2\alpha}})
\nonumber \\
&\cdot \Big(\Gamma_{\rho,n}(
(\rho^{\otimes n})^{\frac{1-\alpha}{2\alpha}}
\sigma^{\otimes n}
(\rho^{\otimes n})^{\frac{1-\alpha}{2\alpha}})\Big)^{\alpha-1} \nonumber \\
\stackrel{(a)}{\le} &
d_n^{1-\alpha}
\Tr 
((\rho^{\otimes n})^{\frac{1-\alpha}{2\alpha}}
\sigma^{\otimes n}
(\rho^{\otimes n})^{\frac{1-\alpha}{2\alpha}})
\nonumber \\
& \cdot ((\rho^{\otimes n})^{\frac{1-\alpha}{2\alpha}}
\sigma^{\otimes n}
(\rho^{\otimes n})^{\frac{1-\alpha}{2\alpha}})^{\alpha-1}\nonumber \\
=&
d_n^{1-\alpha}
\Tr 
((\rho^{\otimes n})^{\frac{1-\alpha}{2\alpha}}
\sigma^{\otimes n}
(\rho^{\otimes n})^{\frac{1-\alpha}{2\alpha}})^{\alpha}\nonumber \\
=&
d_n^{1-\alpha}
(\Tr 
(\rho^{\frac{1-\alpha}{2\alpha}}
\sigma
\rho^{\frac{1-\alpha}{2\alpha}})^{\alpha})^n,\Label{NMR2}
\end{align}
where Step $(a)$ follows from \eqref{NTD5}.

Combining \eqref{NMR1} and \eqref{NMR2}, we have
\begin{align}
& \Tr \rho^{\otimes n} T_{n,\alpha,r}
\le 
d_{n}^2
(\max_{\sigma \in S} 
(\Tr 
(\rho^{\frac{1-\alpha}{2\alpha}}
\sigma
\rho^{\frac{1-\alpha}{2\alpha}})^{\alpha})^{\frac{n}{1-\alpha}}
e^{n\frac{\alpha r}{1-\alpha} } \nonumber \\
=&
d_{n}^2
\exp \Big(n \frac{-(1-\alpha) E_{\alpha}^* (\rho) + \alpha r}{1-\alpha}\Big)\Label{NBVE}.
\end{align}

Since an arbitrary number $\alpha \in [0,1]$ satisfies \eqref{MM1},
\eqref{NBVE} yields
\begin{align}
&-\frac{1}{n}\log \beta_{e^{-nr},n}(\rho)
\nonumber \\
\ge &
-\frac{1}{n}\log d_{n}^2+
\max_{0 \le \alpha \le 1}\frac{(1-\alpha) E_{\alpha}^* (\rho) - \alpha r}{1-\alpha}.
\end{align}
Hence, using \eqref{BVT}, 
we obtain \eqref{NBTX} and \eqref{NBTY}.

\begin{remark}
Although the above presented proof employs
several techniques that come from prior works, 
this proof is not simple combination of 
techniques that come from prior works.
Our test is given in \eqref{BSO},
and it has not been introduced in existing references
\cite{Bjelakovic,Notzel}.
In fact, the method by \cite{Bjelakovic,Notzel} cannot be applied to the case
when $H_0$ hypothesis is the set of all separable states over $n$ bipartite systems, as opposed to an IID $n$-fold tensor product of a separable state.
To extend quantum Sanov theorem to this case, 
our obtained evaluation is essentially needed.
To upper bound the error probability,
we have invented a bounding method different from 
the existing bounding method for the simple hypothesis \cite{simple-hypo}
because the existing bounding method does not work in this case and it is related to
Petz relative R\'{e}nyi entropy.
In fact, the above evaluation requires
new careful handling the parameter $\alpha$
due to the matrix convexity.
\end{remark}

\subsection{Proof of \eqref{NBT2Y}}
Here, we employ the idea invented in \cite{Nagaoka}, whose detail is explained in \cite[Section 3.8]{H-text}.
We choose a test $T_n$ such that
\begin{align}
 \Tr \sigma^{\otimes n}T_n \ge e^{-nr}
\Label{NMA}
\end{align}
for $\sigma \in S$.
For $\alpha \ge 1$ and $\sigma \in S$,
the information processing inequality implies
\begin{align}
& D_{\alpha}(\sigma^{\otimes n}\|\rho^{\otimes n})
\nonumber \\
\ge &
\frac{1}{\alpha-1}\log 
\Big((\Tr \sigma^{\otimes n} T_n)^{\alpha}
(\Tr \rho^{\otimes n} T_n)^{1-\alpha}
\nonumber \\
&+
(\Tr \sigma^{\otimes n} (I-T_n))^{\alpha}
(\Tr \rho^{\otimes n} (I-T_n))^{1-\alpha}
\Big) \nonumber \\
\ge &
\frac{1}{\alpha-1}\log 
\Big((\Tr \sigma^{\otimes n} T_n)^{\alpha}
(\Tr \rho^{\otimes n} T_n)^{1-\alpha}
\Big) \nonumber \\
=&
\frac{\alpha}{\alpha-1}\log \Tr \sigma^{\otimes n} T_n
-\log \Tr \rho^{\otimes n} T_n \nonumber \\
\ge &-\frac{\alpha}{\alpha-1}nr
-\log \Tr \rho^{\otimes n} T_n.
\end{align}
This inequality is converted to
\begin{align}
-\frac{1}{n}\log  \Tr \rho^{\otimes n} T_n
\le 
\frac{1}{n} D_{\alpha}(\sigma^{\otimes n}\|\rho^{\otimes n})
+\frac{\alpha}{\alpha-1}r
\Label{BVT8}.
\end{align}
Since any $T_n$ satisfying \eqref{NMA} satisfies 
\eqref{BVT8}, we have
\begin{align}
&-\frac{1}{n}\log \beta_{1-e^{-nr},n}(\rho)
\le 
\frac{1}{n} D_{\alpha}(\sigma^{\otimes n}\|\rho^{\otimes n})
+\frac{\alpha r}{\alpha-1} \nonumber \\
=& D_{\alpha}(\sigma\|\rho)
+\frac{\alpha r}{\alpha-1}
\Label{BVT2}.
\end{align}
Since any $\sigma\in S$ satisfies \eqref{BVT2}, we have
\begin{align}
&-\frac{1}{n}\log \beta_{1-e^{-nr},n}(\rho)
\le 
\min_{\sigma \in S} D_{\alpha}(\sigma\|\rho)
+\frac{\alpha r}{\alpha-1}
\nonumber \\
=&
E_{\alpha}^*(\rho)
+\frac{\alpha r}{\alpha-1}
\Label{BVT3},
\end{align}
which implies \eqref{NBT2Y}.

\section{Proof of Theorem \ref{TH7}}\Label{AP5}
\subsection{Proof of \eqref{NBT3}}
We employ notations and results
in Appendix \ref{AP4} when $\cH=\cH_A \otimes \cH_B$
and $\rho=\rho_{AB}$.
We define $\Gamma_{\rho_{AB},n}$
and $d_n$ in the same way.

We define the map ${\cal T}_n$ as
${\cal T}_n(\sigma):=\sum_{g \in \frS_n}
\frac{1}{|\frS_n|} \pi(g)\sigma\pi(g)^\dagger$
for a state on $\cH^{\otimes n}$.
We define the set of permutation-invariant separable states as
\begin{align}
{\cal S}_{AB}^{n,sym}:=
\{\sigma \in {\cal S}_{AB}^n| \pi(g)\sigma\pi(g)^\dagger =\sigma
, \quad g \in \frS_n\}.
\end{align}

Then, we have the following lemma.
\begin{lemma}
\begin{align}
&\overline{\beta}_{\epsilon,n}(\rho_{AB})
\nonumber \\
=&
\min_{0\le T\le I}
\Big\{\Tr T \rho_{AB}^{\otimes n} \Big|
\max_{\sigma_{AB} \in {\cal S}_{AB}^{n,sym}}
\Tr(I-T)\sigma_{AB}\le \epsilon
\Big\}.
\end{align}
\end{lemma}

\begin{proof}
Since ${\cal S}_{AB}^{n,sym} \subset {\cal S}_{AB}^{n}$,
we have
\begin{align}
&\min_{0\le T\le I}
\Big\{\Tr T \rho_{AB}^{\otimes n} \Big|
\max_{\sigma_{AB} \in {\cal S}_{AB}^{n} }
\Tr(I-T)\sigma_{AB}\le \epsilon
\Big\}\nonumber \\
\ge &
\min_{0\le T\le I}
\Big\{\Tr T \rho_{AB}^{\otimes n} \Big|
\max_{\sigma_{AB} \in {\cal S}_{AB}^{n,sym} }
\Tr(I-T)\sigma_{AB}\le \epsilon
\Big\}.
\end{align}
Hence, it is sufficient to prove the opposite inequality.

We choose the test $T_{\epsilon}$ such that
any element $\sigma_{AB}' \in {\cal S}_{AB}^{n,sym}$ satisfies 
\begin{align}
&\Tr (I-T_{\epsilon}) \sigma_{AB}' \le \epsilon, \\
&\Tr T_{\epsilon} \rho_{AB}^{\otimes n}
\nonumber \\
 =&
\min_{0\le T\le I}
\Big\{\Tr T \rho_{AB}^{\otimes n} \Big|
\max_{\sigma_{AB} \in {\cal S}_{AB}^{n,sym} }
\Tr(I-T) \sigma_{AB}\le \epsilon
\Big\}.
\end{align}
Due to the symmetric property, 
any element $\sigma_{AB}' \in {\cal S}_{AB}^{n,sym}$ satisfies 
\begin{align}
& \Tr (I-{\cal T}_n(T_{\epsilon})) \sigma_{AB}' \le \epsilon, \\
&\Tr {\cal T}_n(T_{\epsilon}) \rho_{AB}^{\otimes n}
\nonumber \\
 =&
\min_{0\le T\le I}
\Big\{\Tr T \rho_{AB}^{\otimes n} \Big|
\max_{\sigma_{AB} \in {\cal S}_{AB}^{n,sym} }
\Tr(I-T) \sigma_{AB}\le \epsilon
\Big\}.
\end{align}
For any element $\sigma_{AB} \in {\cal S}_{AB}^{n}$,
we have
$
\Tr (I-{\cal T}_n(T_{\epsilon})) \sigma_{AB}
=\Tr (I-{\cal T}_n(T_{\epsilon})) \pi(g)\sigma_{AB}\pi(g)^\dagger$.
Since ${\cal T}_n(\sigma_{AB})\in {\cal S}_{AB}^{n,sym}$, we have
$
\Tr (I-{\cal T}_n(T_{\epsilon})) \sigma_{AB}
=\Tr (I-{\cal T}_n(T_{\epsilon})) {\cal T}_n(\sigma_{AB})\le \epsilon$.
Thus,
\begin{align}
&\min_{0\le T\le I}
\Big\{\Tr T \rho_{AB}^{\otimes n} \Big|
\max_{\sigma_{AB} \in {\cal S}_{AB}^{n} }
\Tr(I-T)\sigma_{AB}\le \epsilon
\Big\}\nonumber \\
\le &
\Tr {\cal T}_n(T_{\epsilon}) \rho_{AB}^{\otimes n}
= \Tr T_{\epsilon} \rho_{AB}^{\otimes n} \nonumber \\
=&\min_{0\le T\le I}
\Big\{\Tr T \rho_{AB}^{\otimes n} \Big|
\max_{\sigma_{AB} \in {\cal S}_{AB}^{n,sym} }
\Tr(I-T)\sigma_{AB}\le \epsilon
\Big\}.
\end{align}
\end{proof}

Since any element $\sigma \in 
{\cal S}_{AB}^{n,sym}$ is written as the form
\begin{align}
\sigma =
\bigoplus_{\lambda \in Y_d^n}
\rho_\lambda \otimes \rho_{\lambda,mix}
\end{align}
in the same way as \eqref{BBR},
replacing the set $\{\sigma^{\otimes n}|\sigma \in S\}$ 
of states
by ${\cal S}_{AB}^{n,sym}$,
given $r >0$,
we define $T_{n,R}$ and $R_{n,\alpha,r}$.
In the same way as \eqref{NMR1}, we have
\begin{align}
&\Tr\rho_{AB}^{\otimes n}T_{n,\alpha,r}\nonumber \\
\le &
d_n (\max_{\sigma \in {\cal S}_{AB}^{n,sym}}
\Tr \Gamma_{\rho_{AB},n}(\sigma)^{\alpha}
(\rho_{AB}^{\otimes n})^{1-\alpha})^{\frac{1}{1-\alpha}}
e^{n\frac{\alpha r}{1-\alpha}} \Label{NMR1B}
\end{align}
for $\alpha \in [0,1]$.
For $\sigma \in {\cal S}_{AB}^{n,sym}$, 
in the same way as \eqref{NBC1},
we have
\begin{align}
&\Tr \sigma (I-T_{n,\alpha,r})
\le e^{-nr}.
\end{align}
Thus, we have
\begin{align}
 -\frac{1}{n}\log \overline{\beta}_{e^{-nr},n}(\rho_{AB})
\ge -\frac{1}{n}\log  \Tr \rho_{AB}^{\otimes n} T_{n,\alpha,r}.\Label{MM1B}
\end{align}

In the same way as \eqref{NMR2}, 
we have
\begin{align}
& \Tr \Gamma_{\rho_{AB},n}(\sigma)^{\alpha} 
(\rho_{AB}^{\otimes n})^{1-\alpha}
\nonumber \\
\le &
d_n^{1-\alpha}
\Tr 
((\rho_{AB}^{\otimes n})^{\frac{1-\alpha}{2\alpha}}
\sigma(\rho_{AB}^{\otimes n})^{\frac{1-\alpha}{2\alpha}})^{\alpha}.
\Label{NMR2B}
\end{align}
Combining \eqref{NMR1B} and \eqref{NMR2B}, we have
\begin{align}
& \Tr \rho_{AB}^{\otimes n} T_{n,\alpha,r}
\le 
d_{n}^2
\exp \Big( \frac{-(1-\alpha) 
{E}_{\alpha}^* (\rho_{AB}^{\otimes n}) + n \alpha r}{1-\alpha}\Big)\Label{NBVEB}
\end{align}
in the same way as \eqref{NBVE}.

Since an arbitrary real $\alpha \in [0,1]$ satisfies \eqref{MM1},
\eqref{NBVEB} yields
\begin{align}
&-\frac{1}{n}\log \beta_{e^{-nr},n}(\rho_{AB})
\nonumber \\
\ge &
-\frac{1}{n}\log d_{n}^2+
\max_{0 \le \alpha \le 1}\frac{\frac{1}{n}(1-\alpha) 
E_{\alpha}^* (\rho_{AB}^{\otimes n}) - \alpha r}{1-\alpha},
\end{align}
using \eqref{BVT}, we obtain \eqref{NBT3} and
\eqref{NBT3X} by taking the limit.

\subsection{Proof of \eqref{NBT4}}
For $\alpha \ge 1$ and $\sigma \in {\cal S}_{AB}^{n}$,
in the same way as \eqref{BVT2}, we have
\begin{align}
-\frac{1}{n}\log \overline{\beta}_{1-e^{-nr},n}(\rho)
\le 
\frac{1}{n} D_{\alpha}(\sigma\|\rho^{\otimes n})
+\frac{\alpha r}{\alpha-1}
\Label{BVT2T}
\end{align}
for $\sigma \in {\cal S}_{AB}^{n}$.
Since any $\sigma\in {\cal S}_{AB}^{n}$ satisfies \eqref{BVT2T}, we have
\begin{align}
-\frac{1}{n}\log \overline{\beta}_{1-e^{-nr},n}(\rho)
\le &
\min_{\sigma \in {\cal S}_{AB}^{n}} \frac{1}{n}D_{\alpha}(\sigma\|\rho^{\otimes n})
+\frac{\alpha r}{\alpha-1}
\nonumber \\
=&
\frac{1}{n}E_{\alpha}^*(\rho^{\otimes n})
+\frac{\alpha r}{\alpha-1}
\Label{BVT3T},
\end{align}
which implies \eqref{NBT4X} and 
\eqref{NBT4} by taking the limit.

\end{document}